\documentclass[fleqn,acmsmall,screen,nonacm]{acmart}

\usepackage{macros}

\title{Hippogriff: a semantic approach to uniting core and modules}
\subtitle{With specification as second-order generalized algebraic theory equipped with phase distinction}
\author{Owen Lynch}
\affiliation{%
  \institution{University of Oxford}
  \country{United Kingdom}}
\affiliation{%
  \institution{Topos Research UK}
  \country{United Kingdom}}
\author{Sam Staton}
\affiliation{%
  \institution{University of Oxford}
  \country{United Kingdom}}

\begin{document}

\maketitle

\section{Introduction}

\subsection{Motivation}

Abstraction features in statically typed programming languages allow the programmer to work generically with respect to a collection of types equipped with certain operations. The module system in ML-family languages is a classic approach to this, but typeclasses in Haskell, traits in Rust, or the dependent object type calculus in Scala can all be used to similar effect \cite{milner-1997-definition,wadler-1989-how,klabnik-2026-rust,amin-2016-essence}.

Each of these features must solve a similar problem: with more expressive systems that use type variables comes a greater difficulty in checking equality of types. Typically, this problem is solved by a syntactic distinction between the ``abstraction level'' (modules, typeclasses, etc.) and the ``value level.'' However, this necessitates a bifurcation of the language into ``abstraction level'' and ``value level'' programming, typically duplicating features of the language (like functions) into type-level functions (such as generic types) and normal functions, which use completely different syntax.\footnote{We refrain here from using the more precise term ``phase distinction'' out of fear that it may not apply in generality to the gamut of languages (Haskell, Scala, Rust) that we wish to also talk about; the interested reader may refer to \cite{harper-1987-type} for a precise account of a \emph{syntactic} phase distinction.}

Dependent languages like Agda, Idris or Lean take a different approach: ``universes'' (types of types) are built in which allow working with types at the ``value'' level \cite{norell-2007-practical,brady-2017-typedriven,christiansen-2023-functional}. Thus, the distinction between type-level programming and normal programming is eliminated. However, this unification comes at a heavy cost: arbitrary program segments must now be evaluated at typechecking time. This means that design choices like purity and strict termination are forced on the language. While some view this requirement as a feature, at the very least this seems to rule out the use of such a unification between type-level and normal programming for a language which wishes to use general recursion.

In 1ML, Rossberg proposes a rethinking of a ML-style module system to use a unified syntax between the ``core level'' and the ``module level'' \cite{rossberg-2018-1ml}. Syntactically, this looks like a dependent language, however instead of typechecking it as a dependent language it is typechecked by first performing a purely syntax-directed translation to System F, and then typechecking the resulting System F terms. While this translation to System F is good in the sense that it shows that there is ``no extra power'' in the seemingly fancier 1ML syntax, it is problematic from a user perspective, because type errors cannot explain what went wrong in terms of the conceptual model that the user has of the system; the user must essentially ``reconstruct'' in their head the translation in order to understand type-checking errors.

Another approach to a module system with unified syntax is Sterling and Harper's ModTT \cite{sterling-2021-logical}. ModTT and 1ML are both, in a sense, interpretations of Moggi's semantics for Standard ML modules in the Grothendieck construction of an indexed-category treatment of System F\(_\omega\) \cite{moggi-1991-categorytheoretic}. However, while 1ML sticks with the System F semantics, ModTT just takes the indexed category as a jumping-off point and tries to simply understand from first principles how to describe the associated fibration internally to a type theory. The result is the ``synthetic phase distinction,'' which is simply a proposition in the type theory that can be interpreted as the statement ``it is currently type-checking time.'' Then one can postulate that \(A = A'\) if and only if \(A = A'\) under the assumption of this proposition, and then moreover under the assumption of this proposition, all terms at the ``value-level'' are equal. This seems like a strange assumption, but by how we have constructed our type theory, all functions \(\ms{Int} \to \ms{Type}\) must in fact be constant functions, so we never need to resolve the equality of their argument.

In this paper, we present a specification and implementation of Hippogriff,\footnote{The name ``Hippogriff'' was chosen because the language is a hybrid between a ML-style language (often named after hoofed mammals, e.g. OCaml, OxCaml), and a dependent language (often named after birds, e.g. Agda, Rocq).} a language blending the approaches of 1ML and ModTT. Like 1ML and unlike ModTT, Hippogriff can be given semantics in System F. However, unlike 1ML, the type theory for Hippogriff is defined \emph{before} the translation to System F, so it is possible to write a type checker that explains failure to type check in terms of the syntax that the user actually typed in.

The trick is to use the synthetic phase distinction from Sterling and Harper \emph{in the metatheory}, while leaving it out of the actual type theory. This allows the metatheory to ``talk about'' the split-context nature of System F, while retaining a single-context, unified syntax treatment. There is then a very direct connection between the rules we write down in the unified style and the original (quite complex) rules given in the split-context style in Harper, Mitchell and Moggi \cite{harper-1990-higherorder}.

We then can straightforwardly adapt Coquand's algorithm to handle the resulting dependent type theory with phase distinction, and the resulting language ends up being fairly simple to implement.

\subsection{Overview}

The thesis of this paper is that the synthetic methods of Logical Relations as Types can be used to build a practical programming language. As an additional contribution, we show that such a programming language has a much closer connection to System F than previously thought.

This paper splits cleanly into two relatively self-contained halves. The first half of the paper, consists of \S\ref{sec:hippogriff}, an informal presentation of Hippogriff, and \S\ref{sec:implementation}, a discussion of the implementation and its connections to other programming languages. This half is meant to be readable for someone with a background in functional programming, but not necessarily a background in category theory, though the implementation section might require some experience with type theory.

The second half is a mathematical exposition of the theory that the type theory for Hippogriff is built on. We start with a review of relevant background material on mathematical semantics for dependent type theory in \S\ref{sec:preliminaries}. Then in \S\ref{sec:psogats} we describe the meta-framework in which the type theory for Hippogriff is described, and how type theories in this meta-framework relate to split-context type theories like System F. This meta-framework is an extension of Uemura's ``second-order generalized algebraic theories'' that adds a synthetic phase distinction \cite{uemura-2021-abstract}. We give a brief overview of second-order generalized algebraic theories in \S\ref{sec:sogat-intro}. The full type theory for Hippogriff is then included as a supplemental appendix.

\section{Hippogriff} \label{sec:hippogriff}

We begin the paper with an overview of Hippogriff, the programming language whose implementation accompanies this paper. The full type theory for Hippogriff is specified in an appendix, which requires some involved theory developed in \S\ref{sec:psogats}, but the language can be intuited fairly straightforwardly for a user with some background in functional programming.

\subsection{Hippogriff feature by feature}

\subsubsection{Basics}

Hippogriff has basic types and functions, which can be written down in a syntax that should be familiar to functional programmers. For instance, a function which doubles its argument may be written down and used as follows.
\begin{lstlisting}
  def double (x : Int) : Int := x + x
  def quadruple (x : Int) : Int := double (double x)
\end{lstlisting}
This is in fact just a syntactic sugar over a more direct definition using lambdas.
\begin{lstlisting}
  def double : Int -> Int := x => x + x
  def quadruple : Int -> Int := x => x + x
\end{lstlisting}
Functions are first-class and curried, so we can have
\begin{lstlisting}
  def twice (f : Int -> Int) (x : Int) : Int := f (f x)
  def quadruple : Int -> Int = twice double
\end{lstlisting}

\subsubsection{Type definitions}

Types are syntactically just ordinary values in Hippogriff. This means that there is a type \lstinline{Type} whose values are types. In order to resolve the circularity that this implies, in fact the values of \lstinline{Type} are ``small types'' and \lstinline{Type} itself is a ``large type'', so we do not have \lstinline{Type : Type}. There is no ``type of large types.''

For instance, we write the polymorphic identity function as:
\begin{lstlisting}
  def id (a : Type) (x : a) : a := x
\end{lstlisting}
We can make type definitions using \lstinline{def} as well, such as
\begin{lstlisting}
  def Cont (s : Type) (a : Type) : Type := (a -> s) -> s
  def Cont-return (s : Type) (a : Type) (x : a) : Cont s a := f => f x
\end{lstlisting}

\subsubsection{Small record types}
We can declare and use small record types in the following way (later we discuss large record types).
\begin{lstlisting}
  def Pair (a : Type) : Type := sig
    fst : a
    snd : a
  end

  def dup (a : Type) (x : a) : Pair a := struct
    fst := x
    snd := x
  end

  def total (p : Pair Int) : Int := p.fst + p.snd
\end{lstlisting}
The use of \lstinline{sig} means that \lstinline{Pair} is a fresh type that \emph{behaves like} a record type, so if we then declare another type
\begin{lstlisting}
  def MyPair (a : Type) : Type := sig
    fst : a
    snd : a
  end
\end{lstlisting}
it is not the case that \lstinline{Pair Int} is equal to \lstinline{MyPair Int}, even though they behave in the same way.

\subsubsection{Sum types}

Sum types are declared in the same way as record types, and are fresh in the same way as record types.
\begin{lstlisting}
  def Maybe (a : Type) : Type := sum
    'just a
    'nothing
  end
\end{lstlisting}
We call each line in the definition of a sum type a ``variant'', and each variant has a ``tag'' which is the symbol starting with a single quote. A variant may have many arguments or none.

Elements of the sum type are constructed by applying the tag to arguments of the right type, as in the following
\begin{lstlisting}
  def just1 : Maybe Int := 'just 1
\end{lstlisting}

We can use elements of a sum type by matching on them as normal.
\begin{lstlisting}
  def unwrap (a : Type) (m : Maybe a) : a := match m
    'just x => x
    'nothing => abandon
  end
\end{lstlisting}
Here, \lstinline{abandon} is a way to halt execution of the program, similar to \lstinline{panic!} in Rust.

However, there is an important restriction for pattern matching: the result of a pattern match may not be an element of a large type. For instance, the following is disallowed because we are trying to produce elements of \lstinline{Type}, a large type, from a pattern match.
\begin{lstlisting}
  def method : Type := sum
    'land
    'sea
  end

  def paul-revere (m : method) : Type := match m
    'land => Unit
    'sea => Bool
  end
\end{lstlisting}
It turns out that with this restriction in place, there are in fact \emph{no} non-constant functions \lstinline{method -> Type}. In fact, more generally, there are no non-constant functions \lstinline{a -> Type} for any small type \lstinline{a}. This fact is key to making the type theory for Hippogriff decidable.

\subsubsection{Recursion} \label{sec:recursion}

Every top-level binding in Hippogriff is recursive, in the sense that the name of the thing being defined is in scope during the definition. However, there are restrictions to how that name may be recursively invoked, so that definitions like the following may be rejected.
\begin{lstlisting}
  def Dunno (a : Type) : Type := Dunno a
\end{lstlisting}

There are two cases in which the name may be recursively invoked. The first case is inside of a \lstinline{sum ... end} block, such as
\begin{lstlisting}
  def List (a : Type) : Type := sum
    'cons a (List a)
    'nil
  end
\end{lstlisting}

The second case is inside of a call to the \lstinline{rec (...)} primitive. However, this comes with a restriction like \lstinline{match ... end}; it can only produce elements of small types. Therefore,
\begin{lstlisting}
  def map (a : Type) (b : Type) (f : a -> b) (xs : List a) : List b := match
    'cons x xs => 'cons (f x) (rec (map a b f xs))
    'nil => 'nil
  end
\end{lstlisting}
is allowed, but
\begin{lstlisting}
  def Dunno (a : Type) : Type := rec (Dunno a)
\end{lstlisting}
is not. In future versions of Hippogriff, we may automatically infer the placement of \lstinline{rec}, but for now its explicitness is useful to explain the design.

\subsubsection{Large record types} \label{sec:large-record-types}

Hippogriff supports bundling together types and operations on them via large record types (record types where at least one field is large). For instance, here is the large type of a small type equipped with a less-than-or-equal predicate.
\begin{lstlisting}
  theory Monoid := sig
    t : Type
    unit : t
    op : t -> t -> t
  end
\end{lstlisting}
Note that we use the keyword \lstinline{theory} here instead of \lstinline{def}, and we don't have a type annotation. If we had a ``type of large types'' (which would then need to be larger than a large type), this would be equivalent to \lstinline{def Monoid : Theory := ...}.

Elements of \lstinline{Monoid} may be parameters to definitions produced with \lstinline{def}. For instance, we might use it in the following way.
\begin{lstlisting}
  def concat (m : Monoid) (xs : List m.t) : m.t := match xs
    'nil => m.unit
    'cons y ys => m.op y (rec (concat m ys))
  end
\end{lstlisting}

We can produce elements of large record types in the same way as we do elements of small record types, with \lstinline{struct}.
\begin{lstlisting}
  def Or : Monoid := struct
    t := Bool
    unit := 'false
    op := x => y => match x
      'true => 'true
      'false => y
    end
  end
\end{lstlisting}
Note that these behave like normal records in dependent type theory: we have beta reduction so \lstinline{BoolOrd.t = Bool}, and eta-expansion, so \lstinline{BoolOrd} is equal to any other record with the same fields.

No module system would deserve the name if it did not permit \emph{specialization}. In Hippogriff, this looks like the following
\begin{lstlisting}
  theory SemiRing := sig
    t : Type
    plus : Monoid ~ [ .t := t ]
    times : Monoid ~ [ .t := t ]
  end
\end{lstlisting}
An element of \lstinline{Monoid ~ [ .t := a ]} is a struct \lstinline{m} where \lstinline{m.t} is ``anything you like, so long as it is statically equal to \lstinline{a}.'' In the case of a ``purely static'' type like \lstinline{Type}, this is just \lstinline{a}, but if \lstinline{t} were of type \lstinline{Int}, it can be instantiated with any value.

\subsubsection{Large function types}

A large function type is a function type where the codomain is large. Note that all functions in Hippogriff are dependent, but when the domain is small the codomain type cannot actually depend on the domain.

Large function types can be used like functors in ML. For instance, we replicated the classic ``functional map'' module that 1ML gives as their first example as follows.
\begin{lstlisting}
  theory Eq := sig
    t : Type
    eq : t -> t -> Bool
  end

  theory Map := sig
    key : Type
    map : Type -> Type
    empty : (a : Type) -> map a
    add : (a : Type) -> key -> a -> map a -> map a
    lookup : (a : Type) -> key -> map a -> Maybe a
  end

  def LambdaMap (Key : Eq) : Map := struct
    key := Key.t
    map := a => Key.t -> Option a
    empty := _ => x => 'none
    add := _ => k => v => m => x => match (Key.eq k x)
      'true => 'some v
      'false => m x
    end
    lookup := _ => x => m => m x
  end
\end{lstlisting}

\subsection{Further examples}

In general, Hippogriff attempts to treat semantically and uniformly many features that are handled syntactically in other languages. The fact that these are treated semantically means that they are ``automatically compositional'' and this leads to some interesting features.

\begin{example}
  Recall \lstinline{sort} from \S\ref{sec:large-record-types}.
\begin{lstlisting}
  def sort (m : Ord) (xs : List m.t) : List m.t := ...
\end{lstlisting}
In \lstinline{sort} we pass \lstinline{m : Ord} as a function argument. The ability to do this is a feature that was only added to OCaml in version 5.5, but is just a natural consequence of how everything works by default in Hippogriff; it would require work to \emph{not} have it.
\end{example}

\begin{example}
Hippogriff naturally supports mutually recursive modules, at any level of nesting.
\begin{lstlisting}
  theory EqPair := sig
    even : Eq
    odd : Eq
  end

  def EvenAndOdd : EqPair := struct
    even := struct
      t := sum
        'zero
        'succ EvenAndOdd.odd.t
      end
      eq := x => y => match x
        'zero => match y
           'zero => 'true
           'succ _ => 'false
        end
        'succ x0 => match y
           'zero => 'false
           'succ y0 => rec (EvenAndOdd.odd.eq x0 y0)
        end
      end
    end
    odd := struct
      t := sum
        'succ EvenAndOdd.even.t
      end
      eq := x => y => match x
        'succ x0 => match y
           'succ y0 => rec (EvenAndOdd.even.eq x0 y0)
        end
      end
    end
  end
\end{lstlisting}
A full comparison of Hippogriff's recursive modules with recursive module features in OCaml or Standard ML has not been made yet, but we can say with confidence that no other language unifies the handling of recursive functions, recursive types, and recursive modules to this degree.
\end{example}

\section{Implementation} \label{sec:implementation}

The implementation for Hippogriff is included as a supplement to this paper. In this section, we discuss some of the design choices and algorithms in this implementation. In particular, we discuss the features in Hippogriff that are more or less ``off the shelf'' from the dependently typed literature, and then we discuss the deviations from the standard dependently typed setup that permit Hippogriff to work as it does. We finish with a comparison to similar languages.

\subsection{General design}

We start by discussing the general design of Hippogriff, via touchpoints in contemporary programming language research.

\subsubsection{Bidirectional elaboration}
One of the key deviations of Hippogriff from a typical statically typed general purpose functional language is the use of bidirectional elaboration instead of Hindley-Milner unification as the main typechecking algorithm. Bidirectional elaboration has a long history; we refer the reader to Dunfield and Krishnaswami for a comprehensive survey \cite{dunfield-2022-bidirectional}.

The ``elaboration'' part of bidirectional elaboration refers to the fact that the typechecking algorithm looks like \(\ms{elaborate} \colon \ms{Notation} \to \ms{Maybe}\,\ms{Syntax}\), rather than \(\ms{typecheck} \colon \ms{Notation} \to \ms{Bool}\), where ``notation'' is the term that we use to refer to the trees that the parser produces, and ``syntax'' is the term we use to refer to well-typed terms. One way of thinking about elaboration is that it is the ``parse, don't validate'' dictum \cite{king-2019-parse} applied to type checking.

The ``bidirectional'' part of bidirectional elaboration refers to the refinement of \(\ms{elaborate}\) into two mutually-recursive functions
\begin{align*}
  &\ms{syn} \colon \ms{Notation} \to \ms{Maybe}\,(\ms{Syntax}, \ms{Type}) \\
  &\ms{chk} \colon \ms{Notation} \to \ms{Type} \to \ms{Maybe}\,\ms{Syntax}
\end{align*}
The function \(\ms{syn}\) takes in notation and attempts to produce both syntax and a type for that syntax; we call this \emph{synthesis}. The function \(\ms{chk}\) takes in notation and a type and attempts to produce syntax that has that given type; we call this \emph{checking}.

Certain syntactic forms lend themselves naturally to one mode or the other. For instance, when type checking the expression \(f\,x\), we first attempt to \emph{synthesize} a type \(A \to B\) for \(f\), then \emph{check} that \(x\) has the type \(A\), and then finally \emph{synthesize} the type \(B\) for the whole expression.

When an expression naturally synthesizes a type \(A\) and we are checking that expression at the type \(A'\), we must check that \(A = A'\); this is why it is important to be able to decide the equality of types. We discuss our approach to type equality in \S\ref{sec:modification}.

Bidirectional elaboration has an interesting effect on the design of programming languages, because it opens up the opportunity to use notation that is only allowed in ``checking mode.'' For instance, in Hippogriff the syntax \lstinline{'just x} can only be used when checking against a sum type that has a variant with tag \lstinline{'just}. In synthesis mode, \lstinline{'just x} produces an elaboration error. This means that it is not a problem to have multiple sum types with the same tags, and it also maintains the invariant that each top-level definition only binds the name it looks like it is binding. At any point in the program that would be in synthesis mode, but we need to use a check-only notation, we can simply add a type annotation as in \lstinline{'just x : Maybe Int}, however this seems to occur rarely in practice.

\subsubsection{Nominal types}

Nominal types (types whose identity is given by their name, not their definition) are often an afterthought in textbook and paper accounts of type theories. Sum types are often only treated in the binary, anonymous case. However, they are essential in the practice of programming languages for a variety of reasons, such as readability of type errors and typeclass resolution.

In Hippogriff, we take an approach to nominal types inspired by Narya's kinetic and potential terms \cite{shulman-2026-narya}. A full account of this approach is in the appendix describing the type theory, but we give a brief overview here, which requires some familiarity with the jargon of normalization by evaluation to understand.

The core idea is that each top-level binding creates a fresh variable equipped with a \emph{behavior}. In the case of direct definitions like \lstinline{def MyInt : Type := Int}, the behavior of \lstinline{MyInt} is to \emph{become} \lstinline{Int}, and when \lstinline{MyInt} is \emph{reflected} into the domain of values, it just becomes \lstinline{Int}. But in the case of definitions like
\begin{lstlisting}
  def List (a : Type) : Type := sum
    'cons a (List a)
    'nil
  end
\end{lstlisting}
the behavior is to \emph{be described} by \lstinline{a => sum ... end}. Then when \lstinline{List} is reflected into the domain of values, it becomes a \emph{neutral value} equipped with a behavior, and when the neutral value gets eliminators added to its spine, the behavior is similarly eliminated, so that the behavior of \lstinline{List Int} is just \lstinline{sum ... end}.

We are not aware of any publication on this method for nominal types beyond the source code for Narya \cite{shulman-2026-narya}, but it seems to work quite nicely.

\subsubsection{Specialization}

Specialization in Hippogriff is treated with a variant of the \emph{static extent}, integrated into the record type constructor. The static extent is a variant of singleton types in the context of the synthetic phase distinction \cite{sterling-2021-logical}, so it fits within the tradition of treating specializations via singleton types. It is fairly easy to implement specialization because we already have the machinery for checking whether two values are statically equal as this is needed for deciding type equality.

\subsubsection{Recursion}

We gave an intuitive account of Hippogriff's approach to recursion in \S\ref{sec:recursion}; here we discuss the theory and implementation. The formal account of how this is added to the type theory is in the appendix.

The approach to recursion in Hippogriff is via a modality; this approach is often called ``guarded corecursion'' \cite{nakano-2000-modality,atkey-2013-productive}. The idea in guarded corecursion is that instead of the normal fixpoint function \(\ms{fix}_A \colon (A \to A) \to A\), we instead have a guarded fixpoint \(\ms{gfix}_A \colon (\tri A \to A) \to A\), where \(\tri A\) is a modality. The modality can be used to ensure that the fixpoint is always ``productive''; for instance if the fixpoint produces a stream then we can always ask for another element of that stream.

In Hippogriff, each top-level definition is implicitly wrapped in the guarded fixpoint. The ``productivity'' that we care about is that we can always ask for the behavior of a type.

\begin{example}
  Consider the following definition.
  \begin{lstlisting}
  def List (a : Type) : Type := sum
    'cons a (List a)
    'nil
  end
  \end{lstlisting}
  When we ask for the behavior of \lstinline{List a}, we get back its definition, which involves \lstinline{List a} again. Crucially, this ``recursive'' invocation of \lstinline{List a} is \emph{not} unfolded to its definition. We could repeat this process to ``unfold a stream of type definitions'', but each unfolding step terminates in finite time.
\end{example}

When a type is small, we never care about its value at typechecking time, so we can get out of the modality ``for free'' using \lstinline{rec}.

Note that there is no explicit modality in Hippogriff; we can achieve what we need only by the changes to context structure that would be needed to implement the modality. However, it could be interesting to add the modality explicitly, which could allow for ``productive type transformers'' that acted on recursive sum types.

The implementation of the context structure is based on the account of dependent type theory with a modality given in \cite{gratzer-2019-implementing}.

\subsubsection{Impredicativity}

Famously, System F is \emph{impredicative}, because the domain of quantification for a forall-type includes the type being created. In other words, we can apply an element of \(\forall \alpha. \alpha \to \alpha\) to the type \(\forall \alpha. \alpha \to \alpha\) itself.

Interestingly, we have the ability to make Hippogriff impredicative or predicative with a very small change. Namely, if we say that the size of a function type is given by the size of the codomain of that function type, then Hippogriff is impredicative; \lstinline{(a : Type) -> a -> a} is a small type. However, if we say that the size of a function type is given by the maximum size of the domain and codomain, then \lstinline{(a : Type) -> a -> a} is large, and thus cannot be encoded into \lstinline{Type}.

Currently, we have made the impredicative choice, because we build a model out of System F in any case so we might as well. However, conjecturally the predicative choice could lead to a system that could be fully monomorphized during compilation, essentially because we can statically track the implementations of functions with type arguments.

Going further, we could say ``all function types are large'', in which case functions would no longer be first-class values at runtime. This could allow predictable inlining into higher-order functions like \lstinline{map}, and is similar to how Rust works.

\subsection{Modification for the synthetic phase distinction} \label{sec:modification}

The key source of complexity in writing a bidirectional elaborator for a dependent type theory as opposed to a simple type theory is deciding type equality. Coquand discovered an algorithm for doing this that uses \emph{normalization by evaluation} \cite{coquand-1996-algorithm}.

Normalization by evaluation is a method of computing normal forms for a variety of extensions of the lambda calculus. The idea is the following. For a ``closed'' term, (e.g. a term without free variables) in a terminating programming language we may compute its ``normal form'' by simply evaluating it. Normalization by evaluation extends this by also allowing the evaluation of open terms.

To do this, we create fresh ``dummy values'' for each free variable. These dummy values are called ``neutrals''. Evaluation then proceeds as best as it can, but in certain situations it is not possible to perform certain operations on neutral values. For instance, if \(f\) is a neutral value of function type and we try to evaluate \(f\,t\), we will be ``stuck.'' We handle this by extending the definition of neutrals to allow a ``spine'' of operations which are stuck. In the context of a language with functions and records, the operations which can get stuck are function application and record projection. So in general, a neutral might look something like \(f\,\,a\,.x\,.y\,\,b\), where \(a\) and \(b\) represent function applications and \(.x\) and \(.y\) represent field projections.

In a classical normalization by evaluation setup, we check equality by evaluating syntax into values, and then we ``readback'' or ``quote'' these values back into syntax, and use a standard structural equality check on syntax. However, in practice it is more efficient to ``fuse'' the readback operation and the equality check, and simply write an equality function on values directly.

In Hippogriff, we modify the type of values to also include an ``opaque'' value \(\blacksquare\), which represents a value of a small type; a value we may ignore during type checking. Then problematic syntactic constructs like \lstinline{rec} and \lstinline{abandon} can evaluate to \(\blacksquare\). When checking equality of values \(v_1\), \(v_2\), if either is \(\blacksquare\) we can immediately return true, because we know that we are checking values of a small type.

We do have a readback operation in Hippogriff, which puts \(\blacksquare\) back into the syntax, but its only purpose is to print out values in diagnostics. Elaboration only needs values to do the type equality check; the syntax produced by elaboration never contains \(\blacksquare\) because it never goes through an evaluate-then-readback cycle.

We then use a different interpreter to actually run a Hippogriff program. This interpreter completely ignores \emph{type} values, but actually evaluates fixpoints. Erasure of types could be thought of as a different phase distinction; Theocharis and Brady have work in this direction \cite{theocharis-2026-type}.

\subsection{Similar languages}

\subsubsection{Lean}

Similar to Hippogriff, Lean supports general-recursive functions in the context of a dependently typed programming language. Lean achieves this via marking these functions as \lstinline{unsafe}, which then disallows them from appearing in types \cite{christiansen-2023-functional}.

This approach is not suitable in the context of a module system, however, because a module functor produces syntactic dependencies of the types in the produced module on the terms in the argument module. Thus, no unsafe functions could be used in implementing a module that one ever wishes to apply a module functor to.

\subsubsection{Staged metaprogramming with two-level type theory} 

Kovacs has proposed using two-level type theory for staged metaprogramming \cite{kovacs-2022-staged,kovacs-2024-closurefree}. Kovacs's approach is similar in a way to Hippogriff, in that the ``universe of small types'' plays a similar role in both languages. However, Kovacs's approach is opposite to Hippogriff in how it handles equality of potentially non-terminating object-level code. Specifically, Kovacs simply does not have beta-normalization for object-level code. So Kovacs does an underapproximation of object-level code equality while we do an overapproximation. These are both valid design choices; we chose overapproximation because it is the natural choice from the System F semantics.

\section{Mathematical preliminaries} \label{sec:preliminaries}

This section is key background for \S\ref{sec:psogats}. It is typical for a paper to start with a section like this, but we have deferred it so that the non-categorically inclined reader is not deterred at the start.

The content of this section is a review of the mathematical semantics for type theory. We start in \S\ref{sec:cwf} with categories with families, a classic notion of model for dependent type theory. In \S\ref{sec:indexed-cwf} we discuss indexed categories with families, which handle type theories with multiple contexts (such as System F, or typed predicate logic). Finally in \S\ref{sec:sogat-intro} we introduce second-order generalized algebraic theories, which gives a framework in which one can succinctly and correctly write down type formers in the context of structures similar to categories with families.

\subsection{Categories with families} \label{sec:cwf}

The central object of study in mathematical semantics of type theory is the category of contexts and substitutions associated with a type theory. Many type formers of a type theory are described by universal properties in this category. But this category is ill-equipped to discuss definitional equality for types, as types are treated as either objects in this category or objects in various slice categories thereof, and equality of objects in a category is an ill-behaved notion.

Therefore, it is typical to model the collection of types in each context as a \emph{presheaf} over the category of contexts and substitutions. There are a variety of ways of doing this, but we start by discussing one of the oldest, originally due to Dybjer \cite{dybjer-1996-internal}, though we give a presentation with notation following \cite{kovacs-2022-typetheoretic}.

\begin{definition}
  A \defcase{category with families}, often abbreviated to CwF, consists of:
  \begin{itemize}
    \item A category \(\mc{C}\). The objects of \(\mc{C}\) model contexts and the morphisms substitutions.
    \item A presheaf \(\ms{Ty} \colon \mc{C}\op \to \ms{Set}\). For a context \(\Gamma \colon \mc{C}\), the set \(\ms{Ty}\,\Gamma\) models the types that can be judged in \(\Gamma\). For a substitution \(\gamma \colon \Delta \To \Gamma\) and a type \(A \colon \ms{Ty}\,\Gamma\), we write \(A[\gamma] \colon \ms{Ty}\,\Delta\) to denote the action of substitution on \(A\).
    \item A presheaf \(\ms{Tm} \colon (\groth \ms{Ty})\op \to \ms{Set}\), where \(\groth \ms{Ty}\) is the Grothendieck construction of \(\ms{Ty}\). For a context \(\Gamma \colon \mc{C}\) and a type \(A \colon \ms{Ty}\,\Gamma\), the set \(\ms{Tm}\,\Gamma\,A\) models the terms of type \(A\) judged in context \(\Gamma\). For a substitution \(\gamma \colon \Delta \To \Gamma\), type \(A \colon \ms{Ty}\,\Gamma\) and term \(a \colon \ms{Tm}\,\Gamma\,A\), we write \(a[\gamma] \colon \ms{Tm}\,\Delta\,A[\gamma]\) to denote the action of substitution on \(a\).
    \item A terminal object \(\cdot\) in \(\mc{C}\).
    \item A context extension operation, which takes a context \(\Gamma \colon \mc{C}\) and \(A \colon \ms{Ty}\,\Gamma\) and produces a new context \(\Gamma \tri A\). This new context is characterized by requiring that for a context \(\Delta\), \[\mc{C}(\Delta, \Gamma \tri A) \cong (\gamma \colon \Delta \To \Gamma) \times (\ms{Tm}\,\Delta\,A[\gamma])\]
    From this isomorphism applied to \(1_{\Gamma \tri A}\), we can derive the weakening substitution \(\ms{p} \colon \Gamma \tri A \To \Gamma\) and the variable rule \(\ms{q} \colon \ms{Tm}\,(\Gamma \tri A)\,A[\ms{p}]\).
  \end{itemize}
\end{definition}

A CwF models the basic judgment structure of dependent type theory. In order to model various type formers, we add additional structure.

\begin{definition} \label{def:function-types}
  A CwF has dependent function types if we have operations \((\Pi,\Pi/\ms{lam},\Pi/\ms{cons},\Pi/\beta,\Pi/\eta)\) with types
  \begin{align*}
  &\Pi \colon \{\Gamma \colon \mc{C}\} \to (A \colon \ms{Ty}\,\Gamma) \to \ms{Ty}\,(\Gamma \tri A) \to \ms{Ty}\,\Gamma \\
  &(\Pi/\ms{lam},\Pi/\ms{cons},\Pi/\beta,\Pi/\eta) \colon \{\Gamma \colon \mc{C}\} \to \{A \colon \ms{Ty}\,\Gamma\} \to \{B \colon \ms{Ty}\,(\Gamma \tri A)\} \to \\
  &\quad \ms{Tm}\,(\Gamma \tri A)\,B \cong \ms{Tm}\,\Gamma\,(\Pi\,A\,B)
  \end{align*}
  that are appropriately natural with respect to substitution. The reader's discomfort with the vagueness of ``appropriately natural'' and the nuisance of making it appropriately non-vague serves as a motivation for \S\ref{sec:sogat-intro}.
\end{definition}

Interestingly enough, we can also model \emph{simple} type theory with an appropriate restriction to the definition of CwF, due to Castellan, Clairambault and Dybjer \cite{castellan-2021-categories}.

\begin{definition}
  A \defcase{simple category with families}, abbreviated to SCwF, is a category with families where \(\ms{Ty}\) is constant, meaning that \(\ms{Ty}\,\Gamma = \ms{Ty}\,\cdot\) for all \(\Gamma \in \mc{C}\). An \defcase{unityped category with families}, abbreviated to UCwF, is a CwF where \(\ms{Ty}(\Gamma) = 1\) for all \(\Gamma\).
\end{definition}

\begin{proposition} \label{prop:scwf-pi}
  A SCwF with dependent function types is a model of the simply typed lambda calculus.
\end{proposition}

\begin{proof}
  Note that a SCwF with dependent function types in fact only has simple function types, because having \(B \colon \ms{Ty}\,(\Gamma \tri A)\) is not different from having \(B \colon \ms{Ty}\,\cdot\). 
\end{proof}

\subsection{System F as SCwF-indexed SCwF} \label{sec:indexed-cwf}

An intuitive way of describing System F is that it is two copies of the simply typed lambda calculus stuck on top of each other. We can make this precise via the notion of an indexed category with families. However, before we can develop this, we must discuss the category of CwFs.

\begin{definition}
  A \defcase{CwF morphism} from \((\mc{C}, \ms{Ty}_{\mc{C}}, \ms{Tm}_{\mc{C}})\) to \((\mc{D}, \ms{Ty}_{\mc{D}}, \ms{Tm}_{\mc{D}})\) consists of a functor \(F \colon \mc{C} \to \mc{D}\) along with natural transformations \(\ms{Ty}_F \colon \ms{Ty}_{\mc{C}} \To F^\ast(\ms{Ty}_{\mc{D}})\), \(\ms{Tm}_F \colon \ms{Tm}_{\mc{C}} \To (F, \ms{Ty}_F)^\ast(\ms{Tm}_{\mc{D}})\). Let \(\ms{CwF}\) be the category where the objects are CwFs and the morphisms are CwF-morphisms.
\end{definition}

The following definition is due to Dybjer \cite{dybjer-2019-categories}.

\begin{definition} \label{def:cwf-indexed-cwf}
  A \defcase{CwF-indexed CwF} consists of a CwF \(\mc{C}_0\) along with a functor \(\bb{C}_1 \colon \mc{C}_0\op \to \ms{CwF}\). We call \(\mc{C}_0\) the base CwF, and for \(\Gamma \colon \mc{C}_0\) we call \(\bb{C}_1(\Gamma)\) the fiber CwF over \(\Gamma\). If the base CwF is simple, we call it a SCwF-indexed CwF, and if the fiber CwFs are simple, we call it a CwF-indexed SCwF; if both are simple then obviously it is a SCwF-indexed SCwF. Similar variations can be taken with UCwFs.
\end{definition}

There are many sets involved in a CwF-indexed CwF; to clarify notation and as a conceptual aid we write out their dependency as follows.
\begin{align*}
  &\mc{C}_0 \colon \ms{Set} \\
  &({-}\To_0{-}) \colon \mc{C}_0 \to \mc{C}_0 \to \ms{Set} \\
  &\ms{Ty}_0 \colon \mc{C}_0 \to \ms{Set} \\
  &\ms{Tm}_0 \colon (\Gamma_0 \colon \mc{C}_0) \to \ms{Ty}_0\,\Gamma_0 \to \ms{Set} \\
  &\bb{C}_1 \colon \mc{C}_0 \to \ms{Set} \\
  &({-}\To_1{-}) \colon (\Gamma_0 \colon \mc{C}_0) \to \bb{C}_1\,\Gamma_0 \to \bb{C}_1\,\Gamma_0 \to \ms{Set} \\
  &\ms{Ty}_1 \colon (\Gamma_0 \colon \mc{C}_0) \to \bb{C}_1\,\Gamma_0 \to \ms{Set} \\
  &\ms{Tm}_1 \colon (\Gamma_0 \colon \mc{C}_0) \to (\Gamma_1 \colon \bb{C}_1\,\Gamma_0) \to \ms{Ty}_1\,\Gamma_0\,\Gamma_1 \to \ms{Set}
\end{align*}

CwF-indexed CwFs model the judgment structure of split-context type theories.

\begin{example}
  The judgment structure of classical predicate logic can be modeled by a UCwF-indexed SCwF. The base UCwF models variable contexts and substitutions built out of function symbols; the ``U'' represents the fact that classical predicate logic does not have types. The types in the fiber SCwFs are predicates, and the terms represent judgments \(P_1,\ldots,P_n \vdash Q\). Note that the fiber categories end up being preorders.

  It is then natural to relax this to a SCwF-indexed SCwF, which models typed predicate logic, or a CwF-indexed SCwF, which models dependently typed predicate logic \cite{palmgren-2019-categories}.

  This can be thought of as a refinement of Lawvere's notion of hyperdoctrine which treats more carefully the distinction between types and contexts; see Jacobs \cite{jacobs-1999-categorical} for an in-depth account of hyperdoctrines.
\end{example}

\begin{example}
  The judgment structure of System F can also be modeled by a UCwF-indexed SCwF. The presheaves \(\ms{Ty}_0,\ms{Tm}_0,\ms{Ty}_1,\ms{Tm}_1\) model kinds, constructions, types, and terms respectively, as found in Harper, Mitchell, Moggi \cite{harper-1990-higherorder}. Similarly, \(\mc{C}_0\) models kind contexts, and \(\bb{C}_1(\Gamma_0)\) models type contexts in a kind context. In System F, there is only one kind, \(\star\), which is why this is a UCwF.

  If we want to instead model System F\(_\omega\), we change to a SCwF-indexed SCwF, because in addition to the kind \(\star\) we also have \(\star \to \star\), etc.
\end{example}

Now that we have the judgment structure, we can start postulating type formers. First of all, any type former that can be added to a CwF can be imported at either the base or fiber level.

\begin{example}
  In System F\(_\omega\), both base and fiber SCwFs have function types.
\end{example}

However, the more interesting case is when type formers involve both the base and the fiber in a non-trivial way.

\begin{definition} \label{def:universal-quantification}
  A CwF-indexed CwF has universal quantification if we have a tuple of operations \((\forall, \forall/\ms{lam},\forall/\ms{app},\forall/\beta,\forall/\eta)\) with types
  \begin{align*}
    &\forall \colon \{\Gamma_0 \colon \mc{C}_0\} \to (A_0 \colon \ms{Ty}_0\,\Gamma_0) \to \{\Gamma_1 \colon \bb{C}_1\,\Gamma_0\} \to \ms{Ty}_1\,(\Gamma_0 \tri A_0)\,(\Gamma_1[\ms{p}_{A_0}]) \to \ms{Ty}_1\,\Gamma_0\,\Gamma_1 \\
    &(\forall/\ms{lam},\forall/\ms{App},\forall/\beta,\forall/\eta) \colon \{\Gamma_0 \colon \mc{C}_0\} \to (A_0 \colon \ms{Ty}_0\,\Gamma_0) \to \{\Gamma_1 \colon \bb{C}_1\,\Gamma_0\} \to \ms{Ty}_1\,(\Gamma_0 \tri A_0)\,(\Gamma_1[\ms{p}_{A_0}]) \to \\
    & \quad \ms{Tm}_1\,(\Gamma_0 \tri A)\,(\Gamma_1[\ms{p}_{A_0}])\,A_1 \cong \ms{Tm}_1\,\Gamma_0\,\Gamma_1\,(\forall\,A_0\,A_1)
  \end{align*}
  that are ``appropriately natural'' with respect to substitution, where now the difficulty of figuring out what ``appropriately natural'' means is significantly higher than in the case of a single CwF as in Definition~\ref{def:function-types}. This will be handled automatically in \S\ref{sec:system-f}.
\end{definition}

While the structure of a CwF-indexed CwF matches quite directly the syntactic structure of a split-context type theory, this means that it also inherits the inherent complexity of a split-context type theory. The key innovation in this paper is to build a framework which is equivalent in power to a CwF-indexed CwF, but is significantly less cluttered.

\subsection{Second-order generalized algebraic theories} \label{sec:sogat-intro}

As we have mentioned in Definitions~\ref{def:function-types} and \ref{def:universal-quantification}, defining type formers in the CwF framework is somewhat fiddly because one must always remember to add appropriate naturality conditions. Second-order generalized algebraic theories (SOGATs) solve this problem by using the metatheory to keep track of the context rather than explicitly passing context objects around \cite{uemura-2021-abstract, kaposi-2024-secondorder}. One way of thinking about this is that SOGATs formalize the use of higher-order abstract syntax \cite{pfenning-1988-higherorder} for use in a logical framework \cite{pfenning-1999-system}. Another way to think about this is that SOGATs are a common generalization of generalized algebraic theories \cite{cartmell-1986-generalised} and second-order algebraic theories \cite{fiore-1999-abstract}. Generalized algebraic theories add type dependency to algebraic theories and second-order algebraic theories add variable binding; SOGATs add both.

We give a brief introduction to SOGATs, which requires first an introduction to GATs (generalized algebraic theories), but we refer the reader to the theses of Kovacs \cite{kovacs-2022-typetheoretic} and Bocquet \cite{bocquet-2025-relative} for full accounts in a modern style.

\begin{definition}
  A \defcase{generalized algebraic theory} is a category with families finitely presented by the postulation of types, terms, and equalities.
\end{definition}

\begin{definition}
  A \defcase{model} of a GAT \(T\) is a CwF-morphism from \(T\) into \(\ms{Set}\), where \(\ms{Set}\) is equipped with the CwF-structure with
  \begin{align*}
    & \ms{Ty}_{\ms{Set}}\,\Gamma = \Gamma \to \ms{Set} \\
    & \ms{Tm}_{\ms{Set}}\,\Gamma\,A = (\gamma \colon \Gamma) \to A\,\gamma
  \end{align*}
\end{definition}

\begin{example}
  There is a GAT, the models of which are graphs, presented by
  \begin{align*}
    & \ms{Vertex} \colon \ms{Ty}\,\cdot \\
    & \ms{Edge} \colon \ms{Ty}\,(\cdot \tri \ms{Vertex} \tri \ms{Vertex}[\ms{p}_{\ms{Vertex}}])
  \end{align*}
  Recall that \(\ms{Vertex}[\ms{p}_{\ms{Vertex}}]\) is the application of the weakening substitution \(\ms{p}_{\ms{Vertex}} \colon \cdot
  \tri \ms{Vertex} \To \cdot\) to \(\ms{Vertex}\).
\end{example}

It quickly becomes tedious to work in this style, with explicit weakenings and no named variables. We can make our life easier by ``working internally''. Formally speaking we work in a type theory known as the theory of signatures for GATs; a signature in this type theory presents a GAT \cite{kovacs-2022-typetheoretic}.

We also take this opportunity to rename \(\ms{Ty}\) to \(\ms{Jdg}\) in order to emphasize that these ``types'' model meta-level judgments, and also to free up \(\ms{Ty}\) as a name that we can use.

\begin{example}
  The theory of graphs in theory of signatures style may be instead presented by
  \begin{align*}
    & \ms{Vertex} \colon \ms{Jdg} \\
    & \ms{Edge} \colon \ms{Vertex} \to \ms{Vertex} \to \ms{Jdg}
  \end{align*}
\end{example}

This is precisely the same GAT, just written in a different style. The ``functions'' are just handling the context structure.

\begin{example}
  Now that we have a more convenient syntax, we can extend the theory of graphs to the theory of categories.
  \begin{align*}
    & \ms{Ob} \colon \ms{Jdg} \\
    & \ms{Hom} \colon \ms{Ob} \to \ms{Ob} \to \ms{Jdg} \\
    & \ms{id} \colon (A \colon \ms{Ob}) \to \ms{Hom}\,A\,A \\
    & \ms{compose} \colon \{A\,B\,C \colon \ms{Ob}\} \to \ms{Hom}\,A\,B \to \ms{Hom}\,B\,C \to \ms{Hom}\,A\,C \\
    & \ms{unitl} \colon \{A\,B \colon \ms{Ob}\} \to (f \colon \ms{Hom}\,A\,B) \to \ms{compose}\,(\ms{id}\,A)\,f = f \\
    & \ms{unitr} \colon \{A\,B \colon \ms{Ob}\} \to (f \colon \ms{Hom}\,A\,B) \to \ms{compose}\,f\,(\ms{id}\,B) = f \\
    & \ms{assoc} \colon \{A\,B\,C\,D \colon \ms{Ob}\} \to (f \colon \ms{Hom}\,A\,B) \to (g \colon \ms{Hom}\,B\,C) \to (h \colon \ms{Hom}\,C\,D) \to \\
    &\quad \ms{compose}\,(\ms{compose}\,f\,g)\,h = \ms{compose}\,f\,(\ms{compose}\,g\,h)
  \end{align*}
\end{example}

It is even possible to continue on in this style and write down a generalized algebraic theory for CwFs, and implicitly we have been using this fact in our type-theoretic notation for CwFs. However, this requires explicitly managing contexts, which is precisely what we would like to get away from.

One idea for resolving this would be to move from ``finitely presented CwFs'' to ``finitely presented CwFs with dependent function types.'' This allows classic logical framework-style definitions for type theories.

\begin{example}
  Dependent type theory with function types may be presented as a ``free CwF with function types'' in the following way.
  \begin{align*}
    &\ms{Ty} \colon \ms{Jdg} \\
    &\ms{Tm} \colon \ms{Ty} \to \ms{Jdg} \\
    &\Pi \colon (A \colon \ms{Ty}) \to (\ms{Tm}\,A \to \ms{Ty}) \to \ms{Ty} \\
    &(\Pi/\ms{lam},\Pi/\ms{app},\Pi/\beta,\Pi/\eta) \colon \{A \colon \ms{Ty}\} \to \{B \colon \ms{Tm}\,A \to \ms{Ty}\} \to \\
    &\quad ((a \colon \ms{Tm}\,A) \to \ms{Tm}\,(B\,a)) \cong \ms{Tm}\,(\Pi\,A\,B)
  \end{align*}
\end{example}

This is workable, and is similar to the use of free locally cartesian closed categories as a logical framework in \cite{gratzer-2021-syntactic}. However, when defining a type theory in this style, one does not have control over which judgments may appear in the context of the object type theory.

The solution is to have two universes, \(\ms{Jdg}\) and \(\ms{Jdg}^+\). We then only allow dependent functions with domains taken from \(\ms{Jdg}^+\), modeling the fact that only elements of \(\ms{Jdg}^+\) may appear in the context of the object type theory. This is where the name ``second-order'' comes from, as opposed to the ``higher-order'' nature of a CwF with all function types.

\begin{definition}
  A \(\Pi^+\)-CwF is a CwF \(\mc{C}\) equipped with a sub-presheaf \(\ms{Ty}^+ \subset \ms{Ty}\), along with dependent function types when the domain is in \(\ms{Ty}^+\) \cite{bocquet-2025-relative}.
\end{definition}

\begin{definition} \label{def:sogat}
  A \defcase{second-order generalized algebraic theory} is a finitely presented \(\Pi^+\)-CwF, specifically by a signature in an extension of the theory of signatures for GATs that includes another universe \(\ms{Jdg}^+\).
\end{definition}

\begin{example} \label{ex:function-type-sogat}
  Dependent type theory with function types may be presented as a SOGAT in the following way.
  \begin{align*}
    &\ms{Ty} \colon \ms{Jdg} \\
    &\ms{Tm} \colon \ms{Ty} \to \ms{Jdg}^+ \\
    &\Pi \colon (A \colon \ms{Ty}) \to (\ms{Tm}\,A \to \ms{Ty}) \to \ms{Ty} \\
    &(\ms{lam},\ms{app},\Pi/\beta,\Pi/\eta) \colon \{A \colon \ms{Ty}\} \to \{B \colon \ms{Tm}\,A \to \ms{Ty}\} \to \\
    &\quad ((a \colon \ms{Tm}\,A) \to \ms{Tm}\,(B\,a)) \cong \ms{Tm}\,(\Pi\,A\,B)
  \end{align*}
\end{example}

For the reader disinclined to think too hard about categories and presheaves, SOGATs can be simply understood as an alternative notation to the typical ``turnstile'' presentation of a type theory.

If a SOGAT is a finitely presented \(\Pi^+\)-CwFs, the natural story for models of SOGATs is to look for homomorphisms into ``semantic'' \(\Pi^+\)-CwFs. This would bring the model theory for SOGATs into the classic realm of categorical algebra, the original instance of which is of course Lawvere's work on algebraic theories \cite{lawvere-1968-functorial}. To tell this story, however, we shift to a weaker notion: categories with representables \cite{uemura-2021-abstract}.

\begin{definition}
  A \defcase{category with representables} (CwR) is a category \(\mc{P}\) with finite limits equipped with a class \(R_{\mc{P}}\) of morphisms, satisfying the following.
  \begin{itemize}
    \item All morphisms in \(R_{\mc{P}}\) are exponentiable.
    \item \(R_{\mc{P}}\) is stable under pullback.
  \end{itemize}
\end{definition}

Categories with representables are to \(\Pi^+\)-CwFs as finite limit categories are to CwFs; they are further removed from the syntax, but easier to work with mathematically. CwRs were also developed first, as finite limit categories were. Intuitively, the morphisms in \(R_{\mc{P}}\) correspond to the display maps for type families \(A \to \ms{Jdg}^+\).

\begin{example} \label{ex:lcc-cwr}
  Any locally cartesian closed category \(\mc{P}\) can be given CwR structure by letting \(R_{\mc{P}}\) be the class of all morphisms.
\end{example}

If \(\mc{C}\) is a category, then \(\ms{Psh}(\mc{C})\) may be equipped with CwR structure as in Example~\ref{ex:lcc-cwr}. However, we are typically more interested in the structure of Definition~\ref{ex:psh-cwr}, which requires a preliminary definition.

\begin{definition}
  Let \(u \colon \dot{U} \to U\) be a natural transformation between presheaves on \(\mc{C}\). Then \(u\) is \defcase{representable} if for all \(\Gamma \in \mc{C}\), \(A \colon \yo \Gamma \to U\), the pullback of \(A\) along \(u\) is a representable presheaf. In other words, there exists an object \(\Gamma \tri A \in \mc{C}\) such that
\[\begin{tikzcd}
	{\yo (\Gamma \tri A)} & {\dot{U}} \\
	{\yo \Gamma} & U
	\arrow[from=1-1, to=1-2]
	\arrow[from=1-1, to=2-1]
	\arrow["\lrcorner"{anchor=center, pos=0.125}, draw=none, from=1-1, to=2-2]
	\arrow["u", from=1-2, to=2-2]
	\arrow["A"', from=2-1, to=2-2]
\end{tikzcd}\]
\end{definition}

Representable natural transformations were first used in type theory by Awodey in \cite{awodey-2018-natural}, to make the following definition.

\begin{definition}
  A \defcase{natural model} is a category equipped with a representable morphism of presheaves.
\end{definition}

\begin{definition} \label{ex:psh-cwr}
  There is a CwR structure on \(\ms{Psh}(\mc{C})\) where the representable morphisms are precisely the representable natural transformations; when we use \(\ms{Psh}(\mc{C})\) as a CwR this is always the structure that we mean.
\end{definition}

Uemura defines a \emph{type theory} to be a finitely-presented CwR, which is analogous to how we have defined a SOGAT to be a finitely-presented \(\Pi^+\)-CwF. It is perhaps a bit of an overreach to identify the very general word ``type theory'' with such a specific notion; a better choice might be ``second-order sketch'' in analogy with the concept of finite limit sketch.

In any case, given a CwR \(\mc{T}\), which we have derived via some scheme of presentation (sketch or theory of signatures), we can define a model to be a morphism into some concrete CwR.

\begin{definition}
  A \defcase{CwR morphism} from \(\mc{P}_0\) to \(\mc{P}_1\) is a finite limit functor which preserves both the class of representable maps and exponentiation by representable maps.
\end{definition}

\begin{definition} \label{def:cwr-model}
  Let \(\mc{T}\) be a CwR. Then a \defcase{model} is a choice of category \(\mc{C}\) and morphism from \(\mc{T}\) into \(\ms{Psh}(\mc{C})\), where \(\ms{Psh}(\mc{C})\) has been given the CwR structure of Definition~\ref{ex:psh-cwr}.
\end{definition}

It is important to remember, however, that unlike with algebraic theories where a model is a functor into some fixed category, usually \(\ms{Set}\), a model of a CwR consists of a \emph{choice of category} \(\mc{C}\) and then a functor into \(\ms{Psh}(\mc{C})\).

There is a sense of model which is more analogous, howoever. For any finitely presented CwR \(\mc{T}\) it is possible to write down a finite limit theory of ``a category equipped with a model of \(\mc{T}\)''; models of this finite limit theory in \(\ms{Set}\) then correspond to models of \(\mc{T}\) in the sense of Definition~\ref{def:cwr-model}.

Equivalently, in a more syntactic fashion, any SOGAT may be mechanically translated to a GAT which has the same models \cite{kaposi-2024-secondorder}. This is advantageous because it means that SOGATs have initial models, which may be considered as the ``syntax'' of the SOGAT; this is a key advantage of the second-order approach over the higher-order approach.

\section{SOGATs with phase distinctions} \label{sec:psogats}

It is possible to write down a SOGAT for System F, and this is in fact one of the examples given in \cite{kaposi-2024-secondorder}. However, the control over dependency given by the split-context type theory is lost in the SOGAT setting. This means that while it may be true that in the initial model kinds have no dependency on terms this is not the case in a general model.

The machinery in this section resolves this issue, by controlling how various judgments in a SOGAT depend on the different parts of a split context using a modality.

\subsection{The open and closed modalities associated with a proposition}

Open and closed modalities have a long history within mathematics under the title of Artin gluing \cite{wraith-1974-artin}. In recent years, they have also been successfully applied in a variety of different ways within programming languages under the heading of ``synthetic phase distinctions''. Logical Relations as Types used a synthetic phase distinction to understand logical relations in the context of module systems \cite{sterling-2021-logical}; Sterling then went on to use a synthetic phase distinction in order to give a semantic proof of normalization for Martin-Löf type theory, which could then be generalized to a proof of normalization for cubical type theory \cite{sterling-2021-thesis}. Soon after this, a synthetic phase distinction was used to develop Calf and its successor Decalf, capturing the relationship between logical behavior and runtime cost \cite{niu-2022-costaware,grodin-2024-decalf,grodin-2026-abstraction}. The semantic proof of normalization for Martin-Löf type theory using the synthetic phase distinction was mechanized successfully in \cite{li-2026-mechanizing}. Recently, Theocharis has used a synthetic phase distinction to formalize \emph{erasure}, which allows a compiler for a dependently typed language to safely ignore some of the arguments to a function which are there ``only for typing reasons'' \cite{theocharis-2026-type}.

One of the key advantages of synthetic phase distinctions is that it is possible, and indeed not even difficult, to have many different phase distinctions each associated to different propositions; this allows a compositional treatment of the various topics handled by phase distinction. We do not do this in the current paper for simplicity's sake, but we point out at various parts where it might be advantageous to have another phase distinction.

We begin our account of the phase distinction via a type-theoretic (e.g., non-categorical) treatment. The categorical treatment may be recovered by assuming we are working internally to some topos.

\begin{definition}
Let $\phi$ be a proposition. Then for a type \(A\), define \(\Op A = \phi \to A\). We call \(\Op\) the \defcase{open modality} associated with \(\phi\). In the case when there are multiple open modalities around, we disambiguate by subscript, as in \(\Op\Sub{\phi}\).
\end{definition}

\begin{proposition}
  \(\Op\) is an idempotent monad.
\end{proposition}

\begin{proof}
  \(\Op\) is the reader monad for \(\phi\), and it is idempotent because by propositionality of \(\phi\), \(\phi \cong \phi \times \phi\), so \(\phi \to (\phi \to A) \cong (\phi \times \phi) \to A \cong \phi \to A\).
\end{proof}

\begin{example} \label{ex:prop-classifier}
  Consider the topos \(\ms{Fam}\) of pairs \(X_0 \colon \ms{Set}\), \(X_1 \colon X_0 \to \ms{Set}\), equivalent to the topos of presheaves on the arrow category. As a topos, \(\ms{Fam}\) interprets dependent type theory. Internally to \(\ms{Set}^\to\) we have a proposition \(\phi\) which may be defined externally by \(\phi_0 = 1\), \(\phi_1({-}) = \emptyset\). Then for a type \(X\), \(\Op(X)\) is externally given by \(\Op(X)_0 = X_0\), \(\Op(X)_1({-}) = 1\).

  Intuitively, an element \(X\) of \(\ms{Fam}\) has some information in the base \(X_0\) and some information in the fibers \(X_1\), and \(\Op(X)\) forgets the information in the fibers while keeping the information in the base.
\end{example}

\begin{definition}
  Let \(\phi\) be a proposition. Then for a type \(A\), define \(\Cl A\) by the quotient-inductive type
  \begin{align*}
    &\mb{data}\quad \Cl A \quad \mb{where} \\
    &\quad \eta^{\Cl} \colon A \to \Cl A \\
    &\quad \ast \colon \S \to \Cl A \\
    &\quad \ms{glue} \colon (a \colon A) \to (s \colon \S) \to \eta^{\Cl}\,a = \ast\,s
  \end{align*}
  Mathematically, this is the pushout of the projections \(A \times \phi \to A\) and \(A \times \phi \to \phi\).
\end{definition}

\begin{example} \label{ex:prop-classifier-closed}
  Using the same setup as in Example~\ref{ex:prop-classifier}, \(\Cl A\) is externally given by
  \begin{align*}
    &(\Cl A)_0 = 1 \\
    &(\Cl A)_1(-) = (a_0 \colon A_0) \times A_1\,a_0
  \end{align*}
  Intuitively, if \(\Op\) ``forgets the fibers'', \(\Cl\) ``pushes everything into the fibers''.
\end{example}

\begin{definition}
  A type \(X\) is said to be \defcase{open-modal} with respect to \(\phi\) if the unit \(\eta^{\Op}_X \colon X \to \Op X\) is a bijection. A type is \defcase{closed-modal} if \(\eta^{\Cl}_X \colon X \to \Cl X\) is a bijection.
\end{definition}

One can see that in the setting of Example~\ref{ex:prop-classifier}, a type \(X\) is open-modal if and only if \(X_1(x)\) is the unit for every \(x \colon X_0\), and closed-modal if and only if \(X_0\) is the unit.

\begin{definition}
  If \(\mc{E}\) is a topos and \(\phi\) a proposition in \(\mc{E}\), then we refer to the class of open-modal objects in \(\mc{E}\) by \(\mc{E}_\phi\) and the class of closed-modal objects in \(\mc{E}\) by \(\mc{E}_{\setminus \phi}\).
\end{definition}

\begin{proposition}
  \(\mc{E}_{\phi}\) and \(\mc{E}_{\setminus \phi}\) are both subtopoi of \(\mc{E}\).
\end{proposition}

\begin{example}
  For \(\ms{Fam}\), both \(\ms{Fam}_{\phi}\) and \(\ms{Fam}_{\setminus \phi}\) are isomorphic to \(\ms{Set}\).
\end{example}

\subsection{Simple type theory}

We work through the process of developing SOGATs with phase distinction in the context of the motivating example of the simply-typed lambda calculus.

\begin{definition}
  The following is the SOGAT signature for the simply typed lambda calculus \cite{kaposi-2024-secondorder}.
  \begin{align*}
    &\ms{Ty} \colon \ms{Jdg} \\
    &\ms{Tm} \colon \ms{Ty} \to \ms{Jdg}^+ \\
    &\ms{Fun} \colon \ms{Ty} \to \ms{Ty} \to \ms{Ty} \\
    &(\ms{Fun}/\ms{lam},\ms{Fun}/\ms{app},\ms{Fun}/\beta, \ms{Fun}/\eta) \colon (A\,B \colon \ms{Ty}) \to (\ms{Tm}\,A \to \ms{Tm}\,B) \to \ms{Tm}\,(\ms{Fun}\,A\,B)
  \end{align*}
\end{definition}

There is something unsatisfying however about the above signature. In particular, in a classical presentation of the simply typed lambda calculus, types are judged without a context, whereas this fact is nowhere recorded in the SOGAT signature.

Castellan, Clairamba and Dybjer modified the CwF framework to handle simply typed lambda calculus by requiring that \(\ms{Ty}\) was constant. We would like to be able to express something similar within SOGATs.

The idea is the following. If \(\mc{C}\) is a ``category of contexts'' then for a presheaf \(P\) on \(\mc{C}\), \(P(\Gamma)\) represents the judgments of type \(P\) that can be made in context \(\Gamma\). Let \(\mc{C}^\bullet\) be the category given by freely adding a single object \(\bullet\) to \(\mc{C}\), along with a morphism \(\bullet \to \Gamma\) for every \(\Gamma \colon \mc{C}\). The idea is that for a presheaf \(P\) on \(\mc{C}^\bullet\), \(P(\bullet)\) represents the judgments of type \(P\) \emph{that are made without a context}.

The reason that we have morphisms \(\bullet \to \Gamma\) for each \(\Gamma\) is to record the fact that, in general, a judgment type \(P\) may have some portion which is context-dependent and some portion which is context-independent. The map \(P(\bullet \to \Gamma) \colon P(\Gamma) \to P(\bullet)\) may be thought of exhibiting \(P(\Gamma)\) as a family displayed over \(P(\bullet)\), separating out the part of the judgment that may be made independently of \(\Gamma\) from the part which depends on \(\Gamma\).

In the case where \(P(\bullet \to \Gamma)\) is always the identity, we can see that \(P\) does not in fact depend on \(\Gamma\), \(P\) is an a-priori context-dependent judgment which has turned out to in fact be context-independent.

We now state this all again, but in greater formality, and then use the open and closed modalities associated with a proposition in order to be able to capture this dependency \emph{synthetically}.

\begin{definition}
  If \(F \colon \mc{D} \to \mc{C}\) is a functor, then the cocomma category \(\mc{C} {\uparrow}F\) is the category where the objects are the coproduct of the objects of \(\mc{C}\) and \(\mc{D}\), and the morphisms are given by
  \begin{align*}
    &(\mc{C}{\uparrow}F)(c_1,c_2) = \mc{C}(c_1,c_2) \\
    &(\mc{C}{\uparrow}F)(d_1,d_2) = \mc{D}(d_1,d_2) \\
    &(\mc{C}{\uparrow}F)(c,d) = \mc{C}(c, F\,d) \\
    &(\mc{C}{\uparrow}F)(d,c) = \emptyset
  \end{align*}
  There are inclusions \(\iota_0 \colon \mc{C} \to \mc{C}{\uparrow}F\) and \(\iota_1 \colon \mc{D} \to \mc{C}{\uparrow}F\) defined in the obvious way.
\end{definition}

\begin{definition}
  For \(\mc{C}\) a category, let \(\mc{C}^\bullet\) be \(1{\uparrow}!_{\mc{C}}\) where \(!_{\mc{C}}\) is the unique functor \(\mc{C} \to 1\). We denote by \(\bullet\) the object in \(\mc{C}^\bullet\) coming from \(1\). Let \(\simple = \yo \bullet\), which is a proposition in \(\mc{Psh}(\mc{C})\); we call \(\simple\) the ``simple proposition''.
\end{definition}

\begin{proposition}
  The natural transformation \(\simple \to 1\) is representable, where \(1\) is the terminal presheaf on \(\mc{C}^\bullet\).
\end{proposition}

\begin{proof}
  First note that while \(\simple\) is by definition a representable presheaf, it is not necessarily the case that the natural transformation \(\simple \to 1\) is representable; a natural transformation \(\yo A \to 1\) is representable if and only if the base category has products with \(A\). But in the category \(\mc{C}^\bullet\), all objects have products with \(\bullet\), namely \(A \times \bullet = \bullet\), so we have our desired result.
\end{proof}

Therefore, working internally to \(\ms{Psh}(\mc{C}^\bullet)\), we have \(\simple\) as a representable proposition (subobject of \(1\)). We now investigate the corresponding open and closed modalities.

\begin{theorem} \label{thm:simple-open-modal}
  For \(P\) a presheaf on \(\ms{Psh}(\mc{C}^\bullet)\), we have \(\Op P \cong \ms{const}\,P(\bullet)\). The open modal presheaves are thus the presheaves isomorphic to constant presheaves.
\end{theorem}

\begin{proof}
  Exponentiation by a representable presheaf \(\yo A\) for a category which has products with \(A\) is given by the formula \((\yo A \to P)(\Gamma) = P(\Gamma \times A)\). The result then follows from the fact \(\Gamma \times \bullet = \bullet\).
\end{proof}

\begin{theorem} \label{thm:simple-closed-modal}
  For \(P\) a presheaf on \(\ms{Psh}(\mc{C}^\bullet)\),
  \begin{align*}
    &(\Cl P)(\bullet) = 1 \\
    &(\Cl P)(\Gamma) = P(\Gamma)
  \end{align*}
  with the evident action on morphisms. Consequently, a closed-modal presheaf is one for which \(P(\bullet) \cong 1\), and the closed-modal presheaves are in equivalence with \(\ms{Psh}(\mc{C})\).
\end{theorem}

\begin{proof}
  We show that with this definition, \(\Cl P\) satisfies the following pushout diagram
\[\begin{tikzcd}
	{\simple \times P} & P \\
	{\simple} & {\Cl P}
	\arrow[from=1-1, to=1-2]
	\arrow[from=1-1, to=2-1]
	\arrow[from=1-2, to=2-2]
	\arrow[from=2-1, to=2-2]
	\arrow["\lrcorner"{anchor=center, pos=0.125, rotate=180}, draw=none, from=2-2, to=1-1]
\end{tikzcd}\]
  Consider some presheaf \(Q\) with maps \(f \colon \simple \to Q\) and \(g \colon P \to Q\) such that the pushout diagram commutes. We wish to define \(h \colon \Cl P \to Q\). For \(\Gamma \neq \bullet\), let \(h_\Gamma \colon (\Cl P)(\Gamma) \to Q(\Gamma)\) be \(g_\Gamma\), as \((\Cl P)(\Gamma) = P(\Gamma)\). Then commutation of the pushout diagram implies that \(g_\bullet \colon P(\bullet) \to Q(\bullet)\) must in fact be constant, so we can define \(h_\bullet \colon 1 \to Q(\bullet)\) to be the constant value of \(g_\bullet\). We are done.
\end{proof}

From all of this, we can make the following definition.

\begin{definition}
  A \defcase{phased category with representables} (PCwR) is a CwR \(\mc{P}\) equipped with a representable proposition \(\simple_{\mc{P}} \into 1\) such that the associated closed modality exists (note that the open modality always exists because representable propositions are exponentiable).
\end{definition}

\begin{definition}
  A \defcase{simple model} of a PCwR \(\mc{P}\) is a category \(\mc{C}\) along with a CwR morphism \(F \colon \mc{P} \to \ms{Psh}(\mc{C}^\bullet)\) that sends \(\simple_{\mc{P}}\) to \(\simple_{\ms{Psh}(\mc{C}^\bullet)}\) and preserves the open and closed modalities.
\end{definition}

The next definition is ``conjectural'' in the sense that there is significant work to be done to fill in the technical details. Such work, in greater generality, is in progress by other authors and so it would not be prudent to duplicate a special case here. However, the basic idea is clear to a reader familiar with SOGATs and phase distinctions.

\begin{definition}
  A \defcase{SOGAT with phase distinction} is a signature in the theory of signatures for SOGATs extended by a proposition and its associated open and closed modalities \(\Op\) and \(\Cl\).
\end{definition}

We may now finally give a proper SOGAT signature for the simply typed lambda calculus.

\begin{definition} \label{def:stlc-simple}
  The simply typed lambda calculus may be presented as a SOGAT with phase distinction via the following signature.
  \begin{align*}
    &\ms{Ty} \colon \ms{Jdg} \\
    &\ms{Ty\text{-}is\text{-}open} \colon \ms{IsEquiv}(\eta^{\Op}_{\ms{Ty}}) \\
    &\ms{Tm} \colon \ms{Ty} \to \ms{Jdg}^+ \\
    &\ms{Tm\text{-}is\text{-}closed} \colon (A \colon \ms{Ty}) \to \ms{IsEquiv}(\eta^{\Cl}_{\ms{Tm}\,A}) \\
    &\ms{Fun} \colon \ms{Ty} \to \ms{Ty} \to \ms{Ty} \\
    &(\ms{lam},\ms{app},\ms{Fun}/\beta, \ms{Fun}/\eta) \colon (A\,B \colon \ms{Ty}) \to (\ms{Tm}\,A \to \ms{Tm}\,B) \to \ms{Tm}\,(\ms{Fun}\,A\,B)
  \end{align*}
\end{definition}

As a convenience, we let \(\ms{Jdg}_{\Op} = (A \colon \ms{Jdg}) \times \ms{IsEquiv}\,(\eta^{\Op}_A)\) be the universe of open-modal judgments, and let \(\ms{Jdg}^+_{\Op}\), \(\ms{Jdg}_{\Cl}\), \(\ms{Jdg}_{\Cl}^+\) be similarly defined, so that we could instead write \(\ms{Ty} \colon \ms{Jdg}_{\Op}\) in the above.

\begin{theorem}
  Simple models of the signature in Definition~\ref{def:stlc-simple} are equivalent to SCwFs with function types.
\end{theorem}

\begin{proof}
  First of all, it is not hard to see that a simple model of the signature
  \begin{align*}
    &\ms{Ty} \colon \ms{Jdg}_{\Op} \\
    &\ms{Tm} \colon \ms{Ty} \to \ms{Jdg}^+_{\Cl}
  \end{align*}
  is an SCwF, following the characterizations of open-modal and closed-modal types in Theorems~\ref{thm:simple-open-modal} and \ref{thm:simple-closed-modal}. Then the result follows from a similar argument to Proposition~\ref{prop:scwf-pi}.
\end{proof}

We can now give an alternative synthetic proof of Proposition~\ref{prop:scwf-pi}, which states that for a SCwF, having function types and having dependent function types is equivalent.

\begin{proof}
  In one direction, obviously having dependent function types implies having regular function types. To prove the other direction, let us work synthetically in the signature of Definition~\ref{def:stlc-simple}. We first synthetically prove that in fact any function \(\ms{Tm}\,A \to \ms{Ty}\) must be constant. Let \(B\) be such a function. Then as \(\ms{Ty} \to \Op\,\ms{Ty}\) is an equivalence, to characterize the behavior of \(B\) it suffices to work under the assumption of \(s \colon \S\). But with \(s \colon \S\) in hand, we then have \(\ms{Tm}\,A \cong 1\), because \(\ms{Tm}\,A\) is closed-modal. Thus, \(B\) is constant.

  This shows we can define \(\Pi \colon (A \colon \ms{Ty}) \to (\ms{Tm}\,A \to \ms{Ty}) \to \ms{Ty}\) in terms of \(\ms{Fun}\), and a similar argument show that \(\ms{lam}\), \(\ms{app}\), etc. for \(\Pi\) may also be defined in terms of the ones for \(\ms{Fun}\).
\end{proof}

\subsection{Split-context type theory} \label{sec:system-f}

Now that we have gone through all that pain in order to synthetically characterize simply typed lambda calculus, the reader may be relieved to know that the method extends straightforwardly to a split context type theory like System F. The idea is that instead of the phase distinction controlling what is judged in no context, the phase distinction controls what is judged in the \emph{kind context}. Of course, as mentioned before the same construction can be used for predicate logic, but we focus on the System F case for concreteness.

All that changes from the simple case to the System F case is the definition for model.

\begin{definition}
  Suppose \(\bb{C}_1 \colon \mc{C}_0\op \to \ms{Cat}\) is an indexed category with terminal objects in the base and fibers. Let \(\pi_0 \colon \mc{C}_1 \to \mc{C}_0\) be the associated fibration, and let \(\mc{C}\) be the cograph of \(\pi_0\); we call \(\mc{C}\) the \defcase{extended total space} of \(\bb{C}_1\).
\end{definition}

Note that the \(\mc{C}^\bullet\) used in the last section is the extended total space of the constant functor \(\mc{C} \colon 1 \to \ms{Cat}\).

\begin{definition}
  For an indexed category \(\bb{C}_1 \colon \mc{C}_0\op \to \ms{Cat}\), let \(\mc{C}\) be the extended total space of \(\bb{C}_1\). Then define \(\base \colon \ms{Psh}(\mc{C})\) by
  \begin{align*}
    &\base(\Gamma_0 \colon \mc{C}_0) = 1 \\
    &\base((\Gamma_0, \Gamma_1) \colon \mc{C}_1) = \emptyset
  \end{align*}
  We call \(\base\) the \defcase{base proposition}.
\end{definition}

\begin{proposition} \label{prop:base-is-representable}
  The base proposition \(\base\) is a representable subobject of 1.
\end{proposition}

\begin{proof}
  Given \(\Gamma \colon \mc{C}\) we must find \(\Gamma'\) to make the following diagram a pullback
  \[\begin{tikzcd}
      {\yo \Gamma'} & \base \\
      {\yo \Gamma} & 1
      \arrow[from=1-1, to=1-2]
      \arrow[from=1-1, to=2-1]
      \arrow["\lrcorner"{anchor=center, pos=0.125}, draw=none, from=1-1, to=2-2]
      \arrow[from=1-2, to=2-2]
      \arrow[from=2-1, to=2-2]
    \end{tikzcd}\]
  In the case \(\Gamma = \Gamma_0 \colon \mc{C}_0\), then we let \(\Gamma' = \Gamma_0\). In the case \(\Gamma = (\Gamma_0, \Gamma_1) \colon \mc{C}_1\), then we let \(\Gamma' = \Gamma_0\).

  Both of these work for the same reason: any presheaf with a map into \(\base\) must have empty fibers over \(\mc{C}_1\).
\end{proof}

Note that \(\base\) is a representable presheaf if and only if \(\mc{C}_0\) has a terminal object, but the natural transformation \(\base \to 1\) is always representable.

\begin{definition}
  A \defcase{indexed model} of a PCwR \(\mc{T}\) consists of an indexed category \(\bb{C}_1 \colon \mc{C}_0\op \to \ms{Cat}\) along with a PCwR morphism from \(\mc{T}\) to \(\ms{Psh}(\mc{C})\), where \(\ms{Psh}(\mc{C})\) is equipped with the PCwR structure from Proposition~\ref{prop:base-is-representable}.
\end{definition}

We can see that an indexed model where \(\mc{C}_0 = 1\) is precisely a simple model.

\begin{definition} \label{def:split-context-signature}
  The basic judgment structure of a split-context type theory is as follows, in the form of a SOGAT with phase distinction. 
  \begin{align*}
    &\ms{Ty}_0 \colon \ms{Jdg}_{\Op} \\
    &\ms{Tm}_0 \colon \ms{Ty}_0 \to \ms{Jdg}_{\Op}^+ \\
    &\ms{Ty}_1 \colon \ms{Jdg}_{\Op} \\
    &\ms{Tm}_1 \colon \ms{Ty}_1 \to \ms{Jdg}_{\Cl}^+
  \end{align*}
  We call this the \defcase{split-context signature}.
\end{definition}

\begin{proposition}
  Indexed models of the signature in Definition~\ref{def:split-context-signature} are in equivalence with CwF-indexed SCwFs (see Definition~\ref{def:cwf-indexed-cwf}).
\end{proposition}

\begin{proof}
  We omit this long and fairly technical proof; the reader however can at least check quickly that the type dependencies of the various sets involved in CwF-indexed SCwF are precisely the same type dependencies found in an indexed model of the split-context signature.
\end{proof}

Now, System F does not have dependent kinds. As mentioned before, System F is a SCwF-indexed SCwF, so Definition~\ref{def:split-context-signature} does not quite restrict enough to get us down to the dependency structure of System F.

We could fix this by moving to a system with two propositions \(\simple\) and \(\base\), where \(\simple \leq \base\), and insisting that \(\ms{Ty}_0\) was open-modal with respect to \(\simple\) while \(\ms{Ty}_1\) was open-modal with respect to \(\base\). There is work in progress towards a theory of ``SOGATs equipped with an arbitrary meet semilattice of propositions'' that would make this trivial to do, and the structure of the models would be presheaves on the collage of a diagram \(\mc{C}_1 \to \mc{C}_0 \to 1\).

However, for our application, this is not necessary. Specifically, in order to make type checking decidable in the context of problematic term-level features, we only need to restrict the dependency of types on terms, not kinds on constructors. Perhaps at a later point it will be desirable to make this further restriction (for instance, if it is useful in compilation to have non-dependent kinds), but perhaps dependent kinds are indeed a desirable feature.

Nevertheless, in the same way that a traditional model of System F forms a model of the SOGAT for System F that has no restrictions on type dependency, a traditional model of System F still forms a model of the signature in Definition~\ref{def:split-context-signature}.

We finish the section by discussing type formers. Just like in the context of CwF-indexed CwFs, we may import various traditional type formers in for \((\ms{Ty}_0, \ms{Tm}_0)\) (``at the base level'') or \((\ms{Ty}_1, \ms{Tm}_1)\) (``at the index level''). The interesting case is type formers that involve both families.

\begin{definition}
  In the context of the split-context signature, an \defcase{indexed-type universe} is a base type \(\star \colon \ms{Ty}_0\) such that \(\ms{Tm}_0 \star \cong \ms{Ty}_1\).
\end{definition}

\begin{definition}
  Universal quantification is given by the following extension of the split-context signature.
  \begin{align*}
    &\forall \colon (A \colon \ms{Ty}_0) \to (\ms{Tm}_0\,A \to \ms{Ty}_1) \to \ms{Ty}_1 \\
    &(\forall/\ms{lam}, \forall/\ms{app}, \forall/\beta, \forall/\eta) \colon \{A \colon \ms{Ty}_0\} \to \{B \colon \ms{Tm}_0\,A \to \ms{Ty}_1\} \to \\
    &\quad ((a \colon \ms{Tm}_0\,A) \to \ms{Tm}_1\,(B\,a)) \cong \ms{Tm}_1(\forall\,A\,B)
  \end{align*}
\end{definition}

System F is then characterized by having an indexed-type universe, universal quantification, and index-level function types. System F\(_\omega\) is all of that along with base-level function types.

\subsection{A module system via the synthetic phase distinction}

The essential idea of the Harper, Mitchell, and Moggi module system \cite{harper-1990-higherorder} is the following.

\begin{definition}
  In the context of the split-context signature, we may make the definitions
  \begin{align*}
    &\ms{Ty} \colon \ms{Jdg}_{\Op} \\
    &\ms{Ty} = (A_0 \colon \ms{Ty}_0) \times (\ms{Tm}_0\,A_0 \to \ms{Ty}_1) \\
    &\ms{Tm} \colon \ms{Ty} \to \ms{Jdg}^+ \\
    &\ms{Tm}\,A = (a_0 \colon \ms{Tm}_0\,A_0) \times \ms{Tm}_1\,(A_1\,a_0)
  \end{align*}
\end{definition}

Here \(\ms{Ty}\) is the type of module signatures, and for \(A\) a module signature \(\ms{Tm}\,A\) is the type of modules on that signature.

We may then create various type formers on \((\ms{Ty},\ms{Tm})\) by using type formers on \((\ms{Ty}_0, \ms{Tm}_0, \ms{Ty}_1, \ms{Tm}_1)\).

\begin{theorem}
  The family \((\ms{Ty}, \ms{Tm})\) has dependent function types if
  \begin{itemize}
    \item \((\ms{Ty}_0, \ms{Tm}_0)\) has dependent function types \((\Pi_0,\Pi_0/\ms{lam}, \Pi_0/\ms{app}, \Pi_0/\beta, \Pi_0/\eta)\),
    \item \((\ms{Ty}_1, \ms{Tm}_1)\) has simple function types \((\ms{Fun}, \ms{Fun}/\ms{lam}, \ms{Fun}/\ms{app}, \ms{Fun}/\beta, \ms{Fun}/\eta)\),
    \item and \((\ms{Ty}_0, \ms{Tm}_0, \ms{Ty}_1, \ms{Tm}_1)\) has universal quantification \((\forall, \forall/\ms{lam}, \forall/\ms{app}, \forall/\beta, \forall/\eta)\).
  \end{itemize}
\end{theorem}

\begin{proof}
  We must create definitions for the terms of the dependent function type signature:
  \begin{align*}
    &\Pi \colon (A \colon \ms{Ty}) \to (\ms{Tm}\,A \to \ms{Ty}) \to \ms{Ty} \\
    &(\ms{lam}, \ms{app}, \Pi/\beta, \Pi/\eta) \colon \{A \colon \ms{Ty}\} \to \{B \colon \ms{Tm}\,A \to \ms{Ty}\} \to \\
    &\quad ((a \colon \ms{Tm}\,A) \to \ms{Tm}\,(B\,a)) \cong \ms{Tm}\,(\Pi\,A\,B)
  \end{align*}
  The essential idea is that a function \(X \to Y\) where \(Y\) is open-modal is equivalent to a function \(\Op\,X \to Y\). This is because under the open modality, \(\Op\,X \cong X\). If \(f \colon X \to Y\) with \(Y\) open-modal, we let \(f^
  {\Op} \colon \Op\,X \to Y\) be the equivalent function. Furthermore, we identify \(\Op(\ms{Tm}\,A)\) with \(\ms{Tm}_0\,A_0\).

  We may now define \(\Pi\) in the following way.
  \begin{align*}
    &(\Pi\,A\,B)_0 = \Pi_0\,A_0\,(a_0 \mapsto (B^{\Op}\,a_0)_0) \\
    &(\Pi\,A\,B)_1\,(f_0 \colon \ms{Tm}_0\,(\Pi\,A\,B)_0) = \forall\,A_0\,(a_0 \mapsto \ms{Fun}\,(A_1\,a_0)\,((B^{\Op}\,a_0)_1\,(\Pi_0/\ms{app}\,f_0\,a_0)))
  \end{align*}
  This is a rather complicated definition, but it actually corresponds precisely to the signature for dependent module functors given in \cite{harper-1990-higherorder}. To understand it, first remember \(B^{\Op} \colon \ms{Tm}_0\,A_0 \to \ms{Ty}\). Therefore
  \begin{align*}
    &(\ms{B}^{\Op}\,a_0)_0 \colon \ms{Ty}_0 \\
    &(B^{\Op}\,a_0)_1 \colon \ms{Tm}_0\,(\ms{B}^{\Op}\,a_0)_0 \to \ms{Ty}_1
  \end{align*}
  It should now be possible for the reader follow the remainder of the definition.

  Defining \(\ms{lam}\) and \(\ms{app}\) is now relatively straightforward.
  \begin{align*}
    &\ms{lam}\,(f \colon (a \colon \ms{Tm}\,A) \to \ms{Tm}\,(B\,a)) = ( \\
    &\quad \Pi_0/\ms{lam}\,(a_0 \mapsto (a \mapsto (f\,a)_0)^{\Op}\,a_0), \\
    &\quad \forall/\ms{lam}\,(a_0 \mapsto \ms{Fun}/\ms{lam}\,(a_1 \mapsto f\,(a_0,a_1)_1))) \\
    &\ms{app}\,(f \colon \ms{Tm}\,(\Pi\,A\,B))\,(a \colon \ms{Tm}\,A) = ( \\
    &\quad \Pi_0/\ms{app}\,f_0\,a_0, \\
    &\quad \ms{Fun}/\ms{app}\,(\forall/\ms{app}\,f_1\,a_0)\,a_1)
  \end{align*}
  We leave it to the reader to check \(\Pi/\beta\) and \(\Pi/\eta\).
\end{proof}

Note that the assumption that \((\ms{Ty}_0, \ms{Tm}_0)\) has dependent function types may be relaxed to the assumption that they just have simple function types if we introduce another modality to make \((\ms{Ty}_0, \ms{Tm}_0)\) simple, and then use Proposition~\ref{prop:scwf-pi}.

One may argue similarly to build \(\Sigma\) types out of record type formers on \((\ms{Ty}_0,\ms{Tm}_0,\ms{Ty}_1, \ms{Tm}_1)\). Indeed, this is the approach that Harper, Mitchell and Moggi use to define their module system; all of these rules are written out explicitly in two-context style in agonizing detail.

But this all rapidly gets rather tedious. Essentially, the point that Sterling and Harper made in Logical Relations as Types is that one can instead simply \emph{start} from the position of \((\ms{Ty} \colon \ms{Jdg}_{\Op}, \ms{Tm} \colon \ms{Ty} \to \ms{Jdg}^+)\), postulate type formers such as \(\Pi\) and \(\Sigma\) \emph{directly}, and then the exercise of justifying those type formers in terms of System F is perhaps theoretically interesting but not necessary in order to get a working type theory and programming language.

In this context, one can add problematic features ``under the closed modality'' and they may be ignored during type checking.

\begin{example}
  If we have \((\ms{Ty} \colon \ms{Jdg}_{\Op}, \ms{Tm} \colon \ms{Ty} \to \ms{Jdg}^+)\), we can postulate an indexed-type universe via
  \begin{align*}
    &\star \colon \ms{Ty} \\
    &\ms{decode} \colon \ms{Tm}\,\star \to \ms{Ty} \\
    &\ms{decode\text{-}is\text{-}closed} \colon (A \colon \ms{Tm}\,\star) \to \ms{IsEquiv}(\eta^{\Cl}_{\ms{decode}\,A})
  \end{align*}
  Then the equality of terms in \(\ms{decode}\,A\) for \(A \colon \ms{Tm}\,\star\) is irrelevant to the decision procedure for equality of types. In particular, this means we can add something like
  \begin{align*}
    &\ms{fix} \colon \{A \colon \ms{Tm}\,\star\} \to (\ms{Tm}\,(\ms{decode}\,A) \to \ms{Tm}\,(\ms{decode}\,A)) \to \ms{Tm}\,(\ms{decode}\,A) \\
    &\ms{fix\text{-}is\text{-}fix} \colon \{A \colon \ms{Tm}\,\star\} \to (f \colon \ms{Tm}\,(\ms{decode}\,A) \to \ms{Tm}\,(\ms{decode}\,A)) \to \ms{fix}\,f = f (\ms{fix}\,f)
  \end{align*}
  to the language, and the consequent undecidability of equality for general terms in the language does not matter for deciding equality of elements of \(\ms{Ty}\), because open-modality means that it suffices to decide equality of elements under the assumption of \(\base\), and under the assumption of \(\base\), \(\ms{Tm}\,(\ms{decode}\,A)\) is contractible.
\end{example}

\section{Conclusion}

\subsection{Summary}

In \S\ref{sec:hippogriff}, we presented a language with a combination of features which common wisdom around programming languages says are incompatible: decidable dependent typing and general recursion. In \S\ref{sec:implementation}, we discussed how implementing such a language was possible.

Then in \S\ref{sec:preliminaries} and \S\ref{sec:psogats} we developed mathematics which justifies theoretically the combination of dependent types and general recursion. This mathematics also provides a connection to System F; a type theory equipped with a synthetic phase distinction can be given semantics via a model of a split-context type theory, using constructions inspired by the classical description of module systems in Harper, Mitchell and Moggi \cite{harper-1990-higherorder}.

\subsection{Future work}

\subsubsection{More interesting dependent types}

We have conservatively stuck to the types in Hippogriff which can be justified by the connection to System F. However, the algorithm we use is much more powerful, in that we can put whatever classic dependent types we wish as large types. This would be like Haskell's DataKind mechanism, but with dependent types, so perhaps similar to the Dependent Haskell proposal \cite{weirich-2017-specification}.

\subsubsection{Ergonomics}

As a minimal language, Hippogriff has some ergonomic rough corners. Some of these ergonomic rough corners could simply be fixed with some syntactic sugar, but some are more fundamental.

For instance, in a typical implementation of a generically typed language, type parameters are automatically instantiated and then solved via unification; the current implementation of Hippogriff uses a bidirectional elaboration with no unification. Additionally, Hippogriff has no subtyping between module signatures, including no subtyping between \lstinline{Monoid ~ [ .t := Int ]} and \lstinline{Monoid}. Fortunately, the toolbox of techniques for dependent type theory implementation is quite large, and can be mined for solutions to these problems.

Kovacs in particular has been pushing the envelope on fast and reliable techniques for unification and implicit variables in dependent type theory \cite{kovacs-2020-elaboration,kovacs-2024-efficient,kovacs-2026-smalltt}. Unification in the Hippogriff context is interesting, because unification can only be used to recover elements of \emph{open-modal} types, as once the opaque element \(\blacksquare\) is around we cannot readback a value to syntax.

For record subtyping, Sterling has reviewed different options for \emph{coherent coercions} between record types in \cite{sterling-2025-labelled}, identifying recent work by Sakaguchi as a promising approach \cite{sakaguchi-2023-refinement}.

\subsubsection{Phase distinction for refinement types}

We have mentioned before that predicate logic is another example of a split-context type theory. One could create a dependent type theory similar to Hippogriff where the ``small types'' were propositions that were erased at runtime. This would be a ``refinement type'' system \cite{freeman-1991-refinement,vazou-2014-refinement} similar to Ghalayini and Krishnaswami's Explicit Refinement Types \cite{ghalayini-2023-explicit}. The modality here would give a synthetic account of the erasure of proofs in Explicit Refinement Types.

It would be interesting to try and combine both phase distinctions, to create a modular language that also had refinement types. Because checking the propositions for refinement types requires a finer-grained notion of equality than what is needed for type checking in Hippogriff, our approach to problematic language features via the opaque value would no longer work here. Further work would be needed to understand how general recursion could work in a context like this.

\subsubsection{Metatheorems via relative induction}

Currently, while Hippogriff's type system has a formal description and mathematical semantics, there is no proof beyond ``proof by implementation'' that the algorithm for deciding type equality actually terminates.

Bocquet has recently developed techniques for short, synthetic proofs of correctness for normalization by evaluation under the heading of ``relative induction'' \cite{bocquet-2025-relative}.

Proving normalization for Hippogriff would require two extensions of relative induction. The first extension would be to integrate the open modality into relative induction, as the result we care about is normalization under the open modality, not general normalization.

The second extension would be integrating guarded recursion into relative induction, so that we could formally tackle the statement that ``asking a type for its behavior always terminates.''

For both of these, a theory of ``modal SOGATs'' is required; the present work serves at one motivation for such a development, but such a theory would be additionally useful for many other applications.

\bibliographystyle{plainurl}
\bibliography{main}

\appendix
\section*{Supplemental appendix: the full type theory for Hippogriff}
In this appendix, we define the type theory for Hippogriff as a SOGAT with phase distinction alongside code samples that typecheck and run in the current Hippogriff implementation. We assume various sugar such as inductive datatypes as arguments (which should be thought of as postulating a \emph{family} of rules, useful for defining n-ary type constructors like records) and implicit arguments. These will be explained as they are used.

\section{Case and the basic judgments}

Nominal types are an essential tool for organizing a code base and for providing readable error messages; any programmer would prefer to see a type error mentioning \lstinline|List Int| rather than \lstinline!fix A. ('nil | 'cons Int A)!. Hippogriff uses an approach to nominality that was learned from the source code to Narya \cite{shulman-2026-narya}.\footnote{What we call ``case'', narya calls ``energy'', and we use nominative/descriptive instead of kinetic/potential.} We start off by parametrizing our standard representable family by \emph{case}. We give \(\ms{Case}\) as an inductive type, but we only ever consider maps out of it, so it does not actually increase the power of our language, it is merely a convenience feature.
\begin{gather*}
  \mb{data} \:\: \ms{Case}\colon \ms{Set} \:\: \mb{where} \\
  \quad \ms{nominative} \colon \ms{Case} \\
  \quad \ms{descriptive} \colon \ms{Case} \\
  \\
  \ms{Ty} \colon \ms{Case} \to \ms{Jdg}_{\Op} \\
  \ms{Ty}_n := \ms{Ty}\,\ms{nominative} \\
  \ms{Ty}_d := \ms{Ty}\,\ms{descriptive} \\
  \ms{El} \colon \ms{Case} \to \ms{Ty}_n \to \ms{Jdg}^+ \\
  \ms{El}_n := \ms{El}\,\ms{nominative} \\
  \ms{El}_d := \ms{El}\,\ms{descriptive}
\end{gather*}
The idea here is that we have a normal type theory for \((\ms{Ty}_n, \ms{El}_n)\). However, at certain points we introduce fresh terms that \emph{behave} in a certain way (in the sense of having certain universal properties). This is done in the following way (we reuse the same notation for types and elements).
\begin{gather*}
  \_\BehavesAs\_ \colon \ms{Ty}_n \to \ms{Ty}_d \to \ms{Jdg} \\
  \_\BehavesAs\_ \colon \{A \colon \ms{Ty}_n\} \to \ms{El}_n\,A \to \ms{El}_d\,A \to \ms{Jdg}^+
\end{gather*}

The idea is that instead of the top-level for Hippogriff being the empty context, the top level is a bunch of bound variables with descriptions, e.g \(x_0 \colon A, \_ \colon x_0 \BehavesAs a_0, \ldots, x_n \colon A_n, \_ \colon x_n \BehavesAs a_n\). In the implementation the descriptions are in fact more tightly bound to the variables, but for the specification it suffices to have the description as auxillary terms.

In order to support definitions which are not nominal, we use the following construct.
\begin{gather*}
  \ms{realize} \colon \{A \colon \ms{Ty}_n\} \to \ms{El}_n\,A \to \ms{El}_d\,A \\
  \ms{realize\text{-}behavior} \colon \{A \colon \ms{Ty}_n\} \to (a_0\,a_1 \colon \ms{El}_n\,A) \to (a_0 \BehavesAs \ms{realize}\,a_1) \to a_0 = a_1
\end{gather*}

\begin{example}
  In the following code snippet, we define \lstinline|MyInt| to be \lstinline|Int| directly, and the elaborator inserts a realize node automatically. Then later on, \lstinline|MyInt| is just an alias of \lstinline|Int|.
  \begin{lstlisting}
  def MyInt : Type := Int

  def MyInt/make (x : Int) : MyInt := x
  \end{lstlisting}
\end{example}

\section{Levels and universes}

Hippogriff is meant to be a modular programming language, so we need a ``type of small types.'' However, it turns out to be convenient to go beyond this and also have a ``type of large types''; this makes it easier to support features like module signatures with arguments. We formalize this in the following way, inspired by Sterling's fuss-free universes \cite{sterling-2025-fussfree}. First, we introduce a poset of levels.
\begin{gather*}
  \mb{data} \:\: \ms{Level}\colon \ms{Set} \:\: \mb{where} \\
  \quad \ms{small} \colon \ms{Level} \\
  \quad \ms{large} \colon \ms{Level} \\
  \quad \ms{top} \colon \ms{Level} \\
  \\
  \_\leq\_ \colon \ms{Level} \to \ms{Level} \to \ms{Prop} \\
  \ms{small} \leq \_ := \top \\
  \ms{large} \leq \ms{small} := \bot \\
  \ms{large} \leq \_ := \top \\
  \ms{top} \leq \ms{small} := \bot \\
  \ms{top} \leq \ms{large} := \bot \\
  \ms{top} \leq \ms{top} := \top
\end{gather*}
Then, we have a judgment which says that a type sits at a certain level. This is functorial in the level, so if \(A@\ms{small}\) then \(A@\ms{large}\).
\begin{gather*}
  \_@\_ \colon \{c \colon\ms{Case}\} \to \ms{Ty}\,e \to \ms{Level} \to \ms{Prop}\\
  \ms{@/functorial} \colon \{c \colon\ms{Case}\} \to \{\ell_0\,\ell_1 \colon \ms{Level}\} \to \ell_0 \leq \ell_1 \to \{A \colon \ms{Ty}\,e\} \to A@\ell_0 \to A@\ell_1
\end{gather*}
Finally, we internalize types at certain levels into terms via universes.
\begin{gather*}
  \mb{data} \:\: \ms{Universe} \colon \ms{Level} \to \ms{Level} \to \ms{Set} \: \: \mb{where} \\
  \quad \ms{Universe/small} \colon \ms{Universe}\,\ms{small}\,\ms{large} \\
  \quad \ms{Universe/large} \colon \ms{Universe}\,\ms{large}\,\ms{top} \\
  \\
  \ms{U} \colon \{\ell_0\,\ell_1 \colon \ms{Level}\} \to \ms{Universe}\,\ell_0\,\ell_1 \to \ms{Ty}_n \\
  \ms{U/level} \colon \{\ell_0\,\ell_1 \colon \ms{Level}\} \to (u \colon \ms{Universe}\,\ell_0\,\ell_1) \to (\ms{U}\,u)@\ell_1 \\
  (\ms{code},\ms{decode},\ms{U}/\beta,\ms{U}/\eta) \colon \{\ell_0\,\ell_1 \colon \ms{Level}\} \to \{u \colon \ms{Universe}\,\ell_0\,\ell_1\} \to \{c \colon\ms{Case}\} \to \\
  \quad (A \colon \ms{Ty}\,e)\times (A@\ell_0) \cong \ms{El}\,c\,(\ms{U}\,u)
\end{gather*}
The two parameters to \(\ms{Universe}\) are the levels that the universe \emph{decodes into} and \emph{lives at}, respectively. So, for instance, \(\ms{U}\,\ms{Universe/small}\) stores codes for small types (hence decodes into small types) and is a large type.

\begin{example}
  In Hippogriff, the appropriate coding and decoding is done automatically by the elaborator. For example, consider the polymorphic identity.
  \begin{lstlisting}
    def id (A : Type) (a : A) : A := a
  \end{lstlisting}
  The term \lstinline|A| of type \lstinline|Type| (which is the syntax for the small universe) is automatically decoded into a type when used on the right hand side of a colon.

  In fact, all syntax in Hippogriff is elaborated as terms; there is no direct syntax for types. This makes the elaborator much easier to write, because elaborating a type is merely checking at a universe. Elements of the large universe are introduced with \lstinline|theory|
  \begin{lstlisting}
  theory Eq := sig
    t : Type
    eq : t -> t -> t
  end\end{lstlisting}
  The only difference between \lstinline|def| and \lstinline|theory| is that \lstinline|def| has a return type annotation. If we allowed \lstinline|Theory| as a term, we could have
  \begin{lstlisting}
  def Eq : Theory := sig
    t : Type
    eq : t -> t -> t
  end\end{lstlisting}
  However, that would either require \lstinline|Theory| to be an element of a further universe (which seems overkill), or mean that type elaboration could not be implemented by checking at a universe, and in any case it seems good to have an indication that a top-level binding is a theory.
\end{example}

We have to also say how decoding interacts with behavior.
\begin{gather*}
  ({\BehavesAs}/\ms{code},{\BehavesAs}/\ms{decode},{\BehavesAs}/\ms{U}/\beta, {\BehavesAs}/\ms{U}/\eta) \colon \{\ell_0\,\ell_1 \colon \ms{Level}\} \to \{u \colon \ms{Universe}\,\ell_0\,\ell_1\} \to \\
  \quad (c \colon \ms{El}_n\,(\ms{U}\,u)) \to (c' \colon \ms{El}_d\,(\ms{U}\,u)) \to (\ms{decode}\,c) \BehavesAs (\ms{decode}\,c') \cong c \BehavesAs c'
\end{gather*}

Finally, we add an axiom that says that the elements of a small type are closed-modal.
\begin{gather*}
  \ms{Small\text{-}dynamic} \colon \{A \colon \ms{Ty}_n\} \to \{c \colon \ms{Case}\} \to A@\ms{small} \to \ms{IsEquiv}\,\eta_{\ms{El}\,c\,A}^{\Cl}
\end{gather*}

\section{Dependent function types}

Function types are more or less the standard thing, adapted to support both nominative and descriptive elements.
\begin{gather*}
  \Pi \colon (A \colon \ms{Ty}_n) \to (\ms{El}_n\,A \to \ms{Ty}_n) \to \ms{Ty}_n \\
  (\ms{lam},\ms{app},\Pi/\beta,\Pi/\eta) \colon \{A \colon \ms{Ty}_n\} \to \{B \colon \ms{El}_n\,A \to \ms{Ty}_n\} \to \{c \colon\ms{Case}\} \to \\
  \quad ((a \colon \ms{El}_n\,A) \to \ms{El}\,c\,(B\,a)) \cong \ms{El}\,c\,(\Pi\,A\,B)
\end{gather*}
We then need some axioms for how application interacts with \(\BehavesAs\).
\begin{gather*}
  ({\BehavesAs}/\ms{lam},{\BehavesAs}/\ms{app},{\BehavesAs}/\Pi/\beta,{\BehavesAs}/\Pi/\eta) \colon \{A \colon \ms{Ty}_n\} \to \{B \colon \ms{El}_n\,A \to \ms{Ty}_n\} \to \{c \colon\ms{Case}\} \to \\
  \quad (f \colon \ms{El}_n\,(\Pi\,A\,B)) \to (f' \colon \ms{El}_d\,(\Pi\,A\,B)) \to
  ((a \colon \ms{El}_n\,A) \to \ms{app}\,f\,a \BehavesAs \ms{app}\,f'\,a) \cong f \BehavesAs f'
\end{gather*}

\begin{example}
  After evaluating the definition
  \begin{lstlisting}
  def Maybe (A : Type) : Type := sum
    'just A
    'nothing
  end\end{lstlisting}
  we have \lstinline|Maybe| bound to a neutral of type \lstinline|Type -> Type| with behavior given by the lambda
  \begin{lstlisting}
  A => sum
    'just A
    'nothing
  end\end{lstlisting}
  Following \({\BehavesAs}/\ms{app}\), the behavior of \lstinline|Maybe Int| is then \lstinline|sum { 'just Int; 'nothing }|.
\end{example}

Finally, we must give an axiom that characterizes the level of \(\Pi\) types. As discussed in the main paper the level of a \(\Pi\) type is the level of its codomain.
\begin{gather*}
  \Pi/\ms{level} \colon \{A \colon \ms{Ty}_n\} \to \{B \colon \ms{El}\,A \to \ms{Ty}_n\} \to \{\ell \colon \ms{Level}\} \to \\
  \quad ((x \colon \ms{El}\,A) \to (B\,a)@\ell) \to (\Pi\,A\,B)@\ell
\end{gather*}

\section{Record types}

Typically in dependent type theory papers, only the rules for binary \(\Sigma\) types are presented and the extension to general records is left as an exercise. Fortunately, one advantage of working in a rich metatheory is that we can do general record types without needing a careful use of imprecise ellipses.

Assume that \(\ms{Name}\) is some fixed set, such as the set of non-empty utf8 strings. In fact, for our purposes, it could well be the unit type; the task of determining from user syntax which field is meant is a task for the elaborator, not for the core language.

We then define records as an internalization of telescopes. Note that the use of the inductive set \(\ms{Names}\) is only in argument position, so this can really be thought of as an collection of axiom families, one for each list of names.
\begin{gather*}
  \mb{data}\:\:\ms{Names} \colon \ms{Set}\:\:\mb{where} \\
  \quad \ms{empty} \colon \ms{Names} \\
  \quad \ms{cons} \colon \ms{Name} \to \ms{Names} \to \ms{Names}
\end{gather*}
\begin{gather*}
  \ms{Tele} \colon \ms{Names} \to \ms{Jdg} \\
  \ms{Tele}\,\ms{empty} := [] \\
  \ms{Tele}\,(\ms{cons}\,x\,\mi{xs}) := [\ms{head} \colon \ms{Ty}_n, \ms{tail} \colon \ms{El}_n\,\ms{head} \to \ms{Tele}\,\mi{xs}] \\
  \\
  \ms{Record} \colon \{\mi{xs} \colon \ms{Names}\} \to \ms{Tele}\,\mi{xs} \to \ms{Ty}_d 
\end{gather*}
\begin{gather*}
  \ms{Elts} \colon (c \colon\ms{Case}) \to \{\mi{xs} \colon \ms{Names}\} \to \ms{Tele}\,\mi{xs} \to \ms{Jdg}^+ \\
  \ms{Elts}\,\_\,\{\ms{empty}\}\,\_ = [] \\
  \ms{Elts}\,c\,\{\ms{cons}\,x\,\mi{xs}\}\,t = [\ms{head} \colon \ms{El}\,c\,t.\ms{head}, \ms{tail} \colon \ms{Elts}\,c\,(t.\ms{tail}\,\ms{head})] \\
  \\
  (\ms{struct},\ms{destruct},\ms{Record}/\beta,\ms{Record}/\eta) \colon \{c \colon\ms{Case}\} \to \{t \colon \ms{Tele}_n\} \to \{A \colon \ms{Ty}_n\} \to \\
  \quad \{A\BehavesAs^{\ms{Ty}}\ms{Record}\,t\} \to \ms{Elts}\,c\,t \cong \ms{El}\,c\,A
\end{gather*}
In this scheme, \(\ms{Record}\) returns a descriptive type; this means that records are only allowed to be nominal. Because of this, the universal property takes in an arbitrary type with a witness that that type behaves like a record. It would not be greatly different to also allow ``anonymous'' record types; before the introduction of the nominative/descriptive system in Hippogriff this was how record types worked.

\begin{example}
  Consider the following code.
  \begin{lstlisting}
  theory PointedType := sig
    t : Type
    pt : t
  end

  def Unit : PointedType := struct
    t := sum
      'tt
    end
    pt := 'tt
  end
  \end{lstlisting}
  This declares a new large type that behaves like a record type, and then an element of that large type whose first field behaves like a singleton, and whose second field is the unique point in that singleton. This crucially relies on the ability to form a descriptive struct; if structs were only nominative then they could not have nominal fields.
\end{example}

We then specify how the eliminator interacts with \(\BehavesAs\).
\begin{gather*}
  {\BehavesAs}/\ms{Elts} \colon \{\mi{xs} \colon \ms{Names}\} \to \{t \colon \ms{Tele}\,\mi{xs}\} \to \ms{Elts}\,\ms{nominative}\,t \to \ms{Elts}\,\ms{descriptive}\,t \to \ms{Jdg}^+ \\
  {\BehavesAs}/\ms{Elts}\,\{\ms{empty}\}\,\_\,\_ = []\\
  {\BehavesAs}/\ms{Elts}\,\{\ms{cons}\,x\,\mi{xs}\}\,e\,e' = [\ms{head} \colon e.\ms{head} \BehavesAs e'.\ms{head}, \ms{tail} \colon {\BehavesAs}/\ms{Elts}\,c.\ms{tail}\,e'.\ms{tail}]\\
  \\
  ({\BehavesAs}/\ms{struct},{\BehavesAs}/\ms{destruct},{\BehavesAs}/\ms{Record}/\beta,{\BehavesAs}/\ms{Record}/\eta) \colon \{c \colon\ms{Case}\} \to \{t \colon \ms{Tele}_n\} \to \\
  \quad \{A \colon \ms{Ty}_n\} \to \{A \BehavesAs \ms{Record}\,t\} \to (a \colon \ms{El}_n\,A) \to (a' \colon \ms{El}_d\,A) \to \\
  \quad {\BehavesAs}/\ms{Elts}\,(\ms{destruct}\,a)\,(\ms{destruct}\,a') \cong a \BehavesAs a'
\end{gather*}

And finally, we specify what level a record type lives at; if all of the fields of a record are at a level then the record is at that level.
\begin{gather*}
  @/\ms{Tele} \colon \{\mi{xs} \colon \ms{Names}\} \to \ms{Tele}\,xs \to \ms{Level} \to \ms{Prop} \\
  @/\ms{Tele} \{\ms{empty}\}\,\_\,\_ = \top \\
  @/\ms{Tele} \{\ms{cons}\,x\,\ms{xs}\}\,t\,\ell = [\ms{head} \colon t.\ms{head}@\ell, \ms{tail} \colon @/\ms{Tele}\,t.\ms{tail}\,\ell] \\
  \\
  \ms{Record}/\ms{level} \colon \{\ell \colon \ms{Level}\} \to \{\mi{xs} \colon \ms{Names}\} \to (t \colon \ms{Tele}\,xs) \to (A \colon \ms{Ty}_n) \to \\
  \quad (A \BehavesAs \ms{Record}\,t) \to @/\ms{Tele}\,t\,\ell \to A@\ell
\end{gather*}

\section{Static extent}

The static extent is the approach to singleton types in the context of the synthetic phase distinction. In practice, we integrate the static extent into the record type, but it is easier to specify it separately.

\begin{gather*}
  \ms{Extent} : (A \colon \ms{Ty}_n) \to (a \colon \ms{El}_n\,A) \to \ms{Ty}_n \\
  (\ms{in}, \ms{out}, \ms{Extent}/\beta, \ms{Extent}/\eta) \colon \{A \colon \ms{Ty}_n\} \to \{a \colon \ms{El}_n\,A\} \to \\
  \quad (a' \colon \ms{El}_n\,A) \times \Op(a = a') \cong \ms{El}_n\,(\ms{Extent}\,A\,a)
\end{gather*}

What this says quite directly is that \(\ms{Extent}\,A\,a\) are the elements of \(A\) that are equal to \(a\) under the open modality.

\section{Sum types}

Sum types differ in several ways from the three types (universes, dependent function types, record types) that we have considered so far. First of all, if we are to sustain the static phase distinction, we must be very careful about levels. Specifically, as small types are statically contractible, we must be careful that we can only eliminate into other statically contractible types.

\begin{example}
  In Hippogriff, it is perfectly valid to have a function \lstinline|LandOrSea -> Type|, where
  \begin{lstlisting}
  def LandOrSea : Type := sum
    'land
    'sea
  end\end{lstlisting}
  However, it is a key invariant of the system that the only functions that one can actually write down of that type are constant. Specifically, a function
  \begin{lstlisting}
  def paul-revere (method : LandOrSea) : Type := match method
    'land => Unit
    'sea => Bool
  end
  \end{lstlisting}
  would be rejected with the error message ``cannot eliminate from a sum type into the large type \lstinline|Type|.''
\end{example}

While this means that Hippogriff is potentially vulnerable to British surprise attacks, without this restriction the interaction between sum types, the ``type judgment is static'' axiom and the ``small types are statically contractible'' judgments would end trivializing the system, in the sense of making all types equal.

With that preamble out of the way, we get to the axioms for sum types. Just like record types were parameterized by a list of names, sum types are parameterized by a shape.
\begin{gather*}
  \mb{data}\:\:\ms{Shape} \colon \ms{Set} \:\: \mb{where} \\
  \quad \ms{empty} \colon \ms{Shape} \\
  \quad \ms{cons} \colon \ms{Name} \to \ms{Nat} \to \ms{Shape} \to \ms{Shape} \\
  \\
  \ms{Ty}_e^{\ms{small}} := (A \colon \ms{Ty}_e) \times (A@\ms{small}) \\
  \\
  \ms{STele} \colon \ms{Nat} \to \ms{Jdg} \\
  \ms{STele}\,0 = [] \\
  \ms{STele}\,(n + 1) = [\ms{head} \colon \ms{Ty}_n^{\ms{small}}, \ms{tail} \colon \ms{STele}\,n]
\end{gather*}
\(\ms{STele}\) stands for simple telescope; unlike a normal telescope the types in this telescope are all small and don't depend on each other. We can then define the input to the sum type former as follows.
\begin{gather*}
  \ms{Variants} \colon \ms{Shape} \to \ms{Jdg} \\
  \ms{Variants}\,\ms{empty} := [] \\
  \ms{Variants}\,(\ms{cons}\,x\,n\,s) := [\ms{head} \colon \ms{STele}\,n, \ms{tail} \colon \ms{Variants}\,s] \\
  \\
  \ms{Sum} \colon \{s \colon \ms{Shape}\} \to \ms{Variants}\,s \to \ms{Ty}_d^{\ms{small}}
\end{gather*}
In order to give the introduction rule for \(\ms{Sum}\), we muse make some auxiliary definitions which describe a well-formed application of a tag in the sum type to its proper arguments, according to a \(\ms{Variants}\) specification.
\begin{gather*}
  \ms{data}\:\:\ms{Tag} \colon \ms{Shape} \to \ms{Set} \:\: \mb{where} \\
  \quad \ms{here} \colon \ms{Tag}\,(\ms{cons}\,x\,n\,s) \\
  \quad \ms{there} \colon \ms{Tag}\,s \to \ms{Tag}\,(\ms{cons}\,x\,s) \\
  \\
  \ms{SElts} \colon \{n \colon \ms{Nat}\} \to \ms{STele}\,n \to \ms{Jdg}^+ \\
  \ms{SElts}\,\{0\}\,\_\,\_ := [] \\
  \ms{SElts}\,\{n+1\}\,e\,t := [\ms{head} \colon \ms{El}_n\,t.\ms{head}, \ms{tail} \colon \ms{SElts}_n\,t.\ms{tail}] \\
  \\
  \ms{Args} \colon \{s \colon \ms{Shape}\} \to (t \colon \ms{Tag}\,s) \to \ms{Variants}\,s \to \ms{Jdg}^+ \\
  \ms{Args}\,\ms{here}\,(\ms{cons}\,x\,n\,\_)\,v := \ms{SElts}\,v.\ms{head} \\
  \ms{Args}\,(\ms{there}\,t)\,(\ms{cons}\,\_\,\_\,s)\,v := \ms{Args}\,t\,s\,v.\ms{tail}
\end{gather*}
Finally, we have the introduction, elimination, and computation rules. These are in a different form from other types that we have considered. This is because in other types the rules formed an isomorphism, but SOGATs do not permit a function \emph{into} a meta-level sum type, only out of one (because a function out of a meta-level sum type is really just a collection of functions, using the universal property). Thus, we cannot have \(\ms{El}_n\,A \cong (t \colon \ms{Tag}\,s) \times \ms{Args}\,t\,v\). Instead, the eliminator takes in a continuation out of \((t \colon \ms{Tag}\,s) \times \ms{Args}\,t\,v\) which is equivalent to providing a function for each tag.
\begin{gather*}
  \ms{tag} \colon \{s \colon \ms{Shape}\} \to \{v \colon \ms{Variants}\,s\} \to \{A \colon \ms{Ty}_n\} \to \{A \BehavesAs \ms{Sum}\,v\} \to \\
  \quad (t \colon \ms{Tag}\,s) \to \ms{Args}\,t\,v \to \ms{El}_n\,A \\
  \ms{match} \colon \{s \colon \ms{Shape}\} \to \{v \colon \ms{Variants}\,s\} \to \{A \colon \ms{Ty}_n\} \to \{A \BehavesAs \ms{Sum}\,v\} \to \\
  \quad \ms{El}_n\,A \to \{B \colon \ms{Ty}_n^{\ms{small}}\} \to ((t \colon \ms{Tag}\,s) \to \ms{Args}\,t\,v \to \ms{El}_n\,B) \to \ms{El}_n\,B \\
  \ms{Sum}/\beta \colon \{s \colon \ms{Shape}\} \to \{v \colon \ms{Variants}\,s\} \to \{A \colon \ms{Ty}_n\} \to \{A \BehavesAs \ms{Sum}\,v\} \to \\
  \quad (B \colon \ms{Ty}_n^{\ms{small}}) \to (f \colon (t \colon \ms{Tag}\,s) \to \ms{Args}\,t\,v \to \ms{El}_n\,B) \to (t \colon \ms{Tag}\,s) \to (a \colon \ms{Args}\,t\,v) \to \\
  \quad \ms{match}\,(\ms{tag}\,t\,a)\,f = f\,t\,a
\end{gather*}

\begin{example}
  We can possibly still give warning about the British as long as it is acceptable to return numerals rather than cardinals. We use \lstinline|>| and \lstinline|=| to indicate the input and output of a command at a REPL.
  \begin{lstlisting}
  def paul-revere (method : LandOrSea) : Int := match method
    'land => 1
    'sea => 2
  end

  > paul-revere 'land
  = 1
  \end{lstlisting}
\end{example}

\section{Abandonment}

In ``The Error Model,'' Joe Duffy argues that there are (at least) two types of errors that one encounters in general-purpose programming \cite{duffy-2016-error}. One type of error indicates the failure of an operation that is known to possibly fail, such as the non-existence of a file whose path the user typed in. The other type of error indicates that an invariant that would have been maintained if the rest of the program were correct has in fact not been maintained. Duffy argues that often the only sensible thing to do in this situation is halt the program as soon as possible, as once an error of this sort is detected the extent of the bug is unknowable and therefore recovery is impossible. Duffy calls this kind of halting ``abandoning.'' Of course, there are situations in which a higher-level process can reset to a known-good state, but in general abandonment is an essential component of software that is not completely statically checked. As with current technology, full verification of every interesting invariant that might appear in a general-purpose program is at least a burden to the developer, if not impossible, abandonment is an essential component of a general-purpose programming language.

Fortunately, unlike in a typical dependently typed language, in Hippogriff abandonment does not compromise type checking, because we restrict abandonment to small types and so abandonment is statically evaluated to the opaque element, just like any other operation.

The rule for abandonment is very simple.
\begin{gather*}
  \ms{abandon} \colon \{A \colon \ms{Ty}_n^{\ms{small}}\} \to \ms{El}_n\,A
\end{gather*}

\begin{example}
  We can define an analogue of Rust's \lstinline|.unwrap()| operation as follows.
  \begin{lstlisting}
  def unwrap (A : Type) (ma : Maybe A) : A := match ma
    'just a => a
    'nothing => abandon
  end
  \end{lstlisting}
\end{example}

\section{Recursion}

This section is relatively more conjectural than the rest. The techniques for recursion described here are implemented in Hippogriff and seem to work, but require a slightly more advanced metatheory to describe. We include it in the current paper as a demonstration of the kind of thing that the phase distinction could enable.

Recursion in Hippogriff is based on \emph{guarded recursion}. To do this, we assume an applicative modality \(\later\) on judgments. Our contexts then look like \(x \colon A, x \BehavesAs f (\ms{pure}\,x)\), where \(f \colon \later A \to A\) is the function elaborated for each top-level definition.

Then we introduce two ways of ``unguarding.'' The first way is by modifying the definition of \(\ms{Variants}\) for sum types to take a delayed simple telescope.
\begin{gather*}
  \ms{Variants} \colon \ms{Shape} \to \ms{Jdg} \\
  \ms{Variants}\,\ms{empty} := [] \\
  \ms{Variants}\,(\ms{cons}\,x\,n\,s) := [\ms{head} \colon \later(\ms{STele}\,n), \ms{tail} \colon \ms{Variants}\,s]
\end{gather*}
This allows recursive use of the top-level definition that we are defining, so long as we are producing a type for one of the constructor arguments in a sum type.
\begin{example}
  Mutual recursion is achieved via recursion at a record type.
  \begin{lstlisting}
  theory EvenAndOdd/typeof := sig
    even : Type
    odd : Type
  end

  def EvenAndOdd : EvenAndOdd/typeof := struct
    even := sum
      'zero
      'succ EvenAndOdd.odd
    end

    odd := sum
      'succ EvenAndOdd.even
    end
  end

  def three : EvenAndOdd.odd := 'succ ('succ ('succ 'zero))
  \end{lstlisting}
\end{example}

The second is we always allow recursion at a small type. Just like with abandonment, this does not impact type checking because any recursion is just evaluated to the opaque element.
\begin{gather*}
  \ms{rec} \colon \{A \colon \ms{Ty}_n^{\ms{small}}\} \to \later (\ms{El}_n\,A) \to \ms{El}_n\,A
\end{gather*}

\begin{example}
  We can define \lstinline|map| for lists in the following way.
  \begin{lstlisting}
  def map (A : Type) (B : Type) (f : A -> B) (xs : List A) : List B := match xs
    'empty => 'empty
    'cons x xs => 'cons (f x) (rec (map A B f xs))
  end

  > map Int Int (n => n + 1) ('cons 1 ('cons 2 'empty))
  = 'cons 2 ('cons 3 'empty)
  \end{lstlisting}
\end{example}

\end{document}